\documentclass{llncs} 

\usepackage[utf8]{inputenc}
\usepackage{amssymb}
\usepackage{amsmath}
\newif\ifdraft
 \draftfalse

\newif\ifarxivversion
\arxivversiontrue

\ifdraft 
    \usepackage{todonotes}
    \usepackage{lineno}
    \linenumbers
\fi 

\usepackage{graphicx}
\usepackage{color}
\usepackage{tikz}
\usepackage{hyperref}
\usepackage{enumitem}  

\newenvironment{romenumerate}{
\begin{enumerate}
\itemsep0pt \parskip0pt \parsep0pt%
 }{\end{enumerate}}

\hypersetup{
    colorlinks=true,       
    linkcolor=black,          
    citecolor=black,        
    filecolor=black,      
    urlcolor=black           
}
\usepackage{soul} 

\graphicspath{{images/}}

\newcommand{\R}{\mathbb{R}}
\newcommand{\Z}{\mathbb{Z}}

\newcommand{\eps}{\varepsilon}

\newcommand\obs{V_n^\infty}
\newcommand\tl{v} 

\ifdraft
\newcommand\ms[1]{\todo[color=green!60!white]{\emph{M\v{S}: #1}}}
\newcommand\vb[1]{\todo[color=red!60!white]{\emph{VB: #1}}}
\newcommand\mo[1]{\todo[color=orange!60!white]{\emph{MO: #1}}}
\newcommand\pv[1]{\todo[color=blue!60!white]{\emph{PV: #1}}}
\else
\newcommand\ms[1]{}
\newcommand\vb[1]{}
\newcommand\mo[1]{}
\newcommand\pv[1]{}
\usepackage[disable]{todonotes}
\fi

\newcommand\resp[1]{\emph{RESPONSE: #1}}

\usepackage{etoolbox}

\newcommand{\lv}[1]{}
\newcommand{\appendixText}{}
\ifarxivversion
\newcommand{\toappendix}[1]{\gappto{\appendixText}{{#1}}}
\else
  \newcommand{\toappendix}[1]{}
\fi

\begin{document}

\title{Non-homotopic Loops with a Bounded Number of Pairwise Intersections\thanks{
M\v{S} was supported by the Czech Science Foundation, grant number \mbox{GJ20-27757Y}, with institutional support RVO:67985807.
VB acknowledges the support of the OP VVV MEYS funded project CZ.02.1.01/0.0/0.0/16\_019/0000765 ``Research Center for Informatics''.
This work was supported by the Grant Agency of the Czech Technical University in Prague, grant \mbox{No.~SGS20/208/OHK3/3T/18}.
PV and MO were supported by project 19-17314J of the Czech Science Foundation.
}}

\author{V\'{a}clav Bla\v{z}ej\inst{1}\orcidID{0000-0001-9165-6280} \and
Michal Opler\inst{2}\orcidID{0000-0002-4389-5807} \and
Matas~\v{S}ileikis\inst{3}\orcidID{0000-0002-6353-9105} \and
Pavel Valtr\inst{4}\orcidID{0000-0002-3102-4166}}

\institute{Faculty of Information Technology, Czech Technical University in Prague, Prague, Czech Republic\\
\and
Computer Science Institute, Charles University, Prague, Czech Republic\\
\email{opler@iuuk.mff.cuni.cz} \and
The Czech Academy of Sciences, Institute of Computer Science, Prague, Czech Republic\\
\email{matas.sileikis@gmail.com} \and
Dept. of Applied Mathematics, Faculty of Mathematics and Physics, Charles University, Prague, Czech Republic
}

\maketitle%

\color{black}

\begin{abstract}
Let $V_n$ be a set of $n$ points in the plane and let $x \notin V_n$.
    An \emph{$x$-loop} is a continuous closed curve not containing any point of $V_n$.
    We say that two $x$-loops are \emph{non-homotopic} if they cannot be transformed continuously into each other without passing through a point of $V_n$.
    For $n=2$, we give an upper bound $e^{O\left(\sqrt{k}\right)}$ on the maximum size of a family of pairwise non-homotopic $x$-loops such that every loop has fewer than $k$ self-intersections and any two loops have fewer than $k$ intersections.
    The exponent $O\big(\sqrt{k}\big)$ is asymptotically tight. The previous upper bound bound $2^{(2k)^4}$ was proved by Pach, Tardos, and T\'oth [\emph{Graph Drawing 2020}].
    We prove the above result by proving the asymptotic upper bound $e^{O\left(\sqrt{k}\right)}$ for a similar problem when $x \in V_n$, and by proving a close relation between the two problems.
    \keywords{graph drawing, non-homotopic loops, curve intersections, plane}
\end{abstract}

\section{Introduction}

The \emph{crossing lemma} bounds the number of edge crossings of a graph drawn in the plane where the graph has $n$ vertices and $m \ge 4n$ edges.
It was proved independently by Ajtai, Chv\'atal, Newborn, and Szemer\'edi~\cite{ACNS} and by Leighton~\cite{Leighton}.
Recently, Pach, Tardos, and T\'oth~\cite{PTT} proved a modification of the crossing lemma for multigraphs with non-homotopic edges.
In the proof, they used a bound on the maximum size of collections of so-called non-homotopic loops.
We focus on improving the bounds for two settings -- one used by the authors of \cite{PTT} and a slightly altered one. We provide an upper bound $e^{O\left(\sqrt{k}\right)}$ with the asymptotically tight exponent $O\big(\sqrt{k}\big)$ in both settings (Theorem~\ref{thm:n_two_selfint}); we also show a new relation between the extremal functions in the two settings (Proposition~\ref{prop:ineqs}).


For an integer $n \ge 1$, let $V_n = \{v_1, \dots, v_n\}$ be a set of $n$ distinct points in the plane $\R^2$. Given $x \in \R^2$, an \emph{$x$-loop} is a continuous function $\ell : [0,1] \to \R^2$ such that $\ell(0) = \ell(1) = x$ and $\ell(t) \not\in V_n$ for $t \in (0,1)$. 
Two $x$-loops $\ell_0, \ell_1$ are \emph{homotopic}, denoted $\ell_0 \sim \ell_1$, if there is a continuous function $H: [0,1]^2 \to \R^2$ (a \emph{homotopy}) such that 
\begin{gather*}
 H(0,t) = \ell_0(t) \quad \text{and} \quad H(1,t) = \ell_1(t) \quad\text{for all } t \in [0,1],\\
 H(s,0) = H(s,1) = x \quad \text{for all } s \in [0,1], \text{ and}\\
 H(s,t) \not\in V_n\quad \text{for all } s,t \in (0,1).
\end{gather*}

In the case when $x \in V_n$, we will, without loss of generality, assume $x = v := v_1$, and refer to $x$-loops as $v$-loops (dropping the subscript for simplicity). Henceforth, when we use the term \emph{$x$-loop}, we will tacitly assume that $x \notin V_n$. 
When $x$ (or $v$) is clear from the context we will also call an~$x$-loop ($v$-loop) simply a \emph{loop}.

A \emph{self-intersection} of a loop $\ell$ is an unordered pair $\{t, u\} \subset (0,1)$ of distinct numbers such that $\ell(t) = \ell(u)$, while an intersection of two loops $\ell_1, \ell_2$ is an \emph{ordered} pair $(t, u) \in (0,1)^2$ such that $\ell_1(t) = \ell_2(u)$.

Given integers $n, k \ge 1$ and $x \notin V_n$ ($v \in V_n$), let $f(n,k)$ (respectively, $g(n,k)$) be the largest number of pairwise non-homotopic $x$-loops (respectively, $v$-loops) such that every loop has fewer than $k$ self-intersections and any two loops have fewer than $k$ intersections.  

Pach, Tardos and T\'oth \cite{PTT} considered $x$-loops (they also added a convenient restriction that no loop passes through $x$, which holds trivially in the setting of $v$-loops).
The quantities $f(n,k)$ and $g(n,k)$ are related by the following inequalities.
  In \cite{BOSV_CALDAM} we proved that for every $n, k \ge 1$ we have
\begin{equation}
\label{eq:gfg}
g(n,k) \le f(n,k) \le g(n+1,k).
\end{equation}
In the current paper we give the following inequality, which allows us to improve the upper bound from \cite{PTT} on $f(n,k)$ by proving an upper bound on $g(n,k)$.
\begin{proposition}\label{prop:ineqs}
  For every $n, k \ge 1$ we have
\begin{equation}
\label{eq:f_less_g}
f(n,k) = O(k^2) \cdot g(n,5k).
\end{equation}
\end{proposition}
Proposition~\ref{prop:ineqs} is proved (with a multiplicative constant of $484$) in Section~\ref{sec:prop:ineqs}.

Pach, Tardos and T\'oth \cite{PTT} showed that for $n \ge 2$
\begin{equation}\label{eq:upper}
  f(n,k) \le 2^{(2k)^{2n}}
\end{equation}
and
\begin{equation}
\label{eq:lower}
  f(n,k) \ge \begin{cases} 
    2^{\sqrt{nk}/3}, \ \ \ \ \ {\rm for}\ n \le 2k,  \\
    (n/k)^{k-1}, \ \ {\rm for}\ n \ge 2k.
  \end{cases}
\end{equation}
 Pach, Tardos and T\'oth \cite{PTT} also proved that if $n = 1$, then there are at most $2k + 1$ non-homotopic loops with fewer than $k$ self-intersections (that is, if we do not bound the number of intersections) implying $f(1,k) \le 2k + 1$.


In our main result we focus on the function $g$ in case $n=2$.

Inequalities \eqref{eq:gfg} and \eqref{eq:upper} imply that $g(2,k) \le 2^{16k^4}$.
After submitting this paper to GD2021, the authors became aware that the latter inequality was not the best at that time. From the proofs in the paper of Juvan, Malni\v{c} and Mohar~\cite{JMM} it follows that $f(2,k) \le k^{C k^2}$ for some absolute constant $C>0$ (this paper focuses on the generality of spaces in which the loops are drawn rather than on quantitative bounds; it also implies good bounds for $f(n,k)$ for any fixed $n$). 

In \cite{BOSV_CALDAM} the authors proved
\[
g(n,k) = 2^{O(k)}.
\]
The following theorem (which we prove in Section~\ref{sec:proof_n_two_selfint}) improves the upper bounds on $g(2,k)$ and $f(2,k)$ significantly. Interestingly, the bound on $g(2,k)$ only uses the restriction on self-intersections (this is not enough for $n \ge 3$), while the restriction on intersections is used only in Proposition~\ref{prop:ineqs}.

\begin{theorem}
  \label{thm:n_two_selfint}
  Let $n = 2$. For any $k$, the size of any collection of non-homotopic $v$-loops with fewer than $k$ self-intersections is $e^{ O(\sqrt{k}) }.$
  In particular $g(2,k) = e^{O(\sqrt{k})}$, and, in view of \eqref{eq:f_less_g}, we have $f(2,k) = e^{O(\sqrt{k})}$.
\end{theorem}
Note that in view of \eqref{eq:lower} the exponent $O\big(\sqrt{k}\big)$ is asymptotically tight.

There is still a huge gap between lower and upper bounds on $f(n,k)$ (and $g(n,k)$) for general $n$; see~\cite{PTT} (also the implicit bounds from \cite{JMM} are probably better for many pairs $(n,k)$).
In the proof of Theorem~\ref{thm:n_two_selfint}, we use several lemmas (Lemmas~\ref{lem:reduced_words}--\ref{lem:expansions}), which might help to narrow this gap, as they provide useful tools and are usually stated for general $n$.

 \subsection*{Acknowledgement}
 We thank the referees for a careful reading and numerous corrections, in particular, one of the referees for pointing out the upper bound by Juvan, Malni\v{c} and Mohar.
\section{Setup and Notation}
\label{sec:setup}
\subsection{Obstacles, equator and gaps}
Depending on the context, we will treat $S := \R^2 \setminus V_n$ either as the plane with $n$ points removed, or as a sphere with $n + 1$ points removed (where $n$ of these points come from the set $V_n = \{ v_1, \dots, v_n\}$ and the last point, denoted by $v_0$, corresponds to the ``point at infinity''). We define $\obs = \{v_0\} \cup V_n$ and refer to the elements of $\obs$ as \emph{obstacles}.

  Given a finite collection of loops, by infinitesimal perturbations, without creating any new intersections or self-intersections, we can ensure that 
  \begin{enumerate}
      \item no two (self-)intersections occur at the same point of $S$,
      \item \label{enum:cross} every (self-)intersection is a \emph{crossing}, that is, one loop ``passes to the other side'' of the other loop (rather than two loops `touching'). 
  \end{enumerate}

      Given a drawing of the loops satisfying the above conditions, we choose a closed simple curve on the sphere which goes through the obstacles $v_0, \dots, v_n$ in this order (for $x$-loops, we choose this curve so that it also avoids $x$). We call this loop the \emph{equator}. 
  Removing the equator from the sphere, we obtain two connected sets, which we arbitrarily name the \emph{northern hemisphere} and the \emph{southern hemisphere}.
  We refer to the $n + 1$ open curves into which the equator is split by excluding points $v_i$ as \emph{gaps}. We assign label $i$ to the gap between $v_i$ and $v_{i+1}$, with indices counted modulo $n+1$. Moreover, when talking about $v$-loops, we treat $v = v_1$ as an additional, special gap, with a label $v$; see Figure~\ref{fig:notation_example}, where the equator is drawn as a circle, for simplicity. 
  
  \begin{figure}[h]
    \centering
    \includegraphics{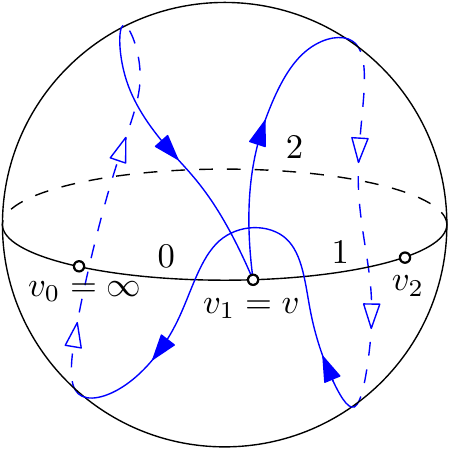}
    \caption{Equator, the gaps, and $v$-loop that induced the word $w=v2102v$.}%
    \label{fig:notation_example}
  \end{figure}
  
  By a careful choice of the equator, we can assume the following conditions: 
  \begin{enumerate}[resume]
      \item 
  every loop in the collection intersects the equator a finite number of times,
  \item 
\label{enum:eq_cross}
each of these intersections (except for, possibly, the intersection at $v$) is a crossing, (i.e., no loop touches the equator),
  \item no point of self-intersection or intersection lies on the equator.
  \end{enumerate}
  \subsection{Segments and induced words}
Part of a given loop $\ell$ between a pair of distinct intersections with the equator (inclusively) is called a \emph{segment}. Treating a loop (respectively, a segment) as a function $\ell: [0,1] \to \infty$ (respectively, the restriction of $\ell$ to a closed subinterval of $[0,1]$), gives a natural orientation of a loop or a segment.
  A minimal segment is called an \emph{arc}.
  If an arc intersects itself, we can remove the part of the arc between these self-intersections without changing the homotopy class of the loop, which allows us to make yet another assumption that 
  \begin{enumerate}[resume]
      \item \label{enum:arc} there are no self-intersections within any arc.
  \end{enumerate}
 

  Consider a segment $s$ that intersects the equator $t$ times (including the beginning and the end). 
  By listing the labels of gaps that the loop crosses as it traverses $s$, we obtain a word $w = w_1\dots w_t$. 
  In this case we say that $s$ is a $w$-\emph{segment}.
  If we take the maximal segment of a loop $\ell$, that is, from the first to the last crossing of equator (which is the whole loop in the case of $v$-loops), then we say that word $w$ is the word \emph{induced} by $\ell$; see Figure~\ref{fig:notation_example}. Given a loop that induces a word $w$, the segments of the loop correspond to subwords (that is, words consisting of consecutive letters of $w$) of length at least $2$.

  Note that the word induced by a $v$-loop $\ell$ starts and ends in $\tl$. Dropping these $\tl$s we obtain a word which we call the \emph{inner word induced by} a $v$-loop $\ell$.

  If we reverse the orientation of a segment $s$, the order of gaps is reversed and hence we obtain the \emph{reverse} of the word $w$, denoted $\overline w := w_t \dots w_1$.	Sometimes we talk about segments as unoriented objects, treating a segment simultaneously as a $w$-segment and a $\overline w$-segment.

  Given an oriented segment $s$, we define the \emph{polarity} of $s$ as the hemisphere to which the first arc of $s$ belongs.
  We call an oriented $w$-segment a \emph{$w$-downsegment} or a \emph{$w$-upsegment} whenever we want to specify the polarity of the first arc. 
  \begin{remark}
    \label{rem:reversing_segments}
    For a word $w$ of even length (that is, with an even number of letters) a $w$-segment is also a $\overline w$-segment of the same polarity, while for a word $w$ of odd length a $w$-segment is also a $\overline w$-segment of the opposite polarity.
  \end{remark}


  For example, consider a $v$-loop with the first arc in the southern hemisphere that induces the word $01201$. It has $01$-segments of both polarities (and hence $10$-segments of both polarities), as well as a $012$-downsegment (which is also a $210$-upsegment) but, say, there is no $012$-upsegment.

  
  \subsection{Patterns}
  To simplify notation we use a concept of \emph{pattern}, that is a finite sequence of symbols, usually, the first letters of the Greek alphabet $\alpha, \beta, \dots$. Given a pattern $T$, we say that a word $w$ is a $T$-word if it is obtained by replacing symbols by letters, so that two letters in $w$ are equal if and only if the symbols in $T$ are equal. For example, words $010, 202, 121$ match the pattern $\alpha\beta\alpha$, but $111$ does not. Given a pattern $T$, we say a word is $T$-free, if it contains no subword matching the pattern $T$.
  With a slight abuse of notation, by a $T$-segment we also mean a $w$-segment such that word $w$ matches the pattern $T$.

\section{General $n$}
In this section we state and prove several facts that are valid for general $n$, including all prerequisites for the proof of Theorem~\ref{thm:n_two_selfint}.

  \subsection{Simplifying words}
  Part \ref{enum:x_is_obst} of the following lemma allows to simplify the words induced by a given family of $v$-loops, in particular simplifying the setting for the proof of Theorem~\ref{thm:n_two_selfint}. Part~\ref{enum:x_not_obst} is used in the proof of the inequality \eqref{eq:f_less_g} of Proposition~\ref{prop:ineqs}.


  \begin{lemma}\label{lem:reduced_words} 
    Assume that $n \ge 1$.
    \begin{romenumerate}
    \item \label{enum:x_is_obst}
      Given any family of $v$-loops, each loop can be replaced by a homotopic $v$-loop inducing an $\alpha\alpha$-free inner word, with the first and last letters in $\{2, \dots, n\}$, so that the numbers of pairwise intersections or self-intersections do not increase.
    \item \label{enum:x_not_obst}
      Suppose that a family of $x$-loops and the equator are such that there is a path connecting $x$ to some $v \in V_n$ which does not intersect any loop or the equator. Then each $x$-loop can be replaced by a homotopic $x$-loop inducing an $\alpha\alpha$-free word so that the numbers of pairwise intersections or self-intersections do not increase.
    \end{romenumerate}
\end{lemma}

The part \ref{enum:x_is_obst} was already proved in \cite{BOSV_CALDAM} and the part \ref{enum:x_not_obst} is easily established using the same ideas.
\ifarxivversion
Therefore, we defer the proof to Appendix~\ref{apx:reduced_words}.
\fi

\toappendix{
\subsection{Lemma~\ref{lem:reduced_words}}\label{apx:reduced_words}
\begin{proof}[Lemma~\ref{lem:reduced_words} \ref{enum:x_is_obst}]
  By an \emph{ear} we mean a segment inducing an $\alpha\alpha$-word.
  Denote $v := v_1$. Taking into account that the gaps $0$ and $1$ are incident to $v$, by \emph{end-ear} we mean a $va$-segment or $av$-segment, where $a \in \{0,1\}$, that is, an end-ear is an initial or final arc of a loop that has one end in the gap $0$ or $1$.
  We will remove ears in the first step (thus deleting the consecutive pairs of equal letters) and end-ears in the second step (thus deleting the $0$\,s and $1$s at the ends of the inner word).

  For the first step, we choose an ear in some loop (between two points of some gap~$a$) and denote its endpoints by $y$ and $z$.
  By \emph{$yz$-gap} denote the set of points in the gap $a$ strictly between $y$ and $z$. An ear is minimal if there is no other ear with both endpoints in the $yz$-gap.
  We remove ears one by one, always picking a minimal ear.
 
  The chosen ear partitions one of the halves of the sphere into two simply connected open sets, one of which, that we denote by $P$, contains the $yz$-gap in its boundary. Let $\overline P$ denote the closure of $P$, that is $P$ together with its boundary.
  
  We remove the chosen ear by continuously transforming it to a path which closely follows the $yz$-gap inside the other hemisphere, as shown on Figure~\ref{fig:ear_removal}. By choosing the new path sufficiently close to the equator we can make sure that if a new \mbox{(self-)intersection} with some loop $\ell$ appears, then by tracing $\ell$ from that \mbox{(self-)intersection} in a certain direction we cross the $yz$-gap, thus entering the set $P$.
  \begin{figure}[h]
    \centering
    \includegraphics[width=.75\textwidth]{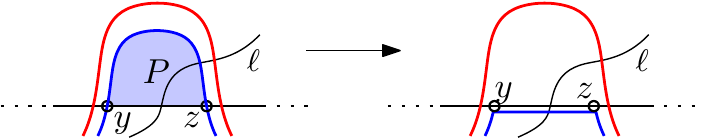}
    \caption{Removal of a minimal ear}%
    \label{fig:ear_removal}
  \end{figure}
  
  Since $v \notin \overline P$, by tracing $\ell$ further we must leave $P$. This cannot happen by crossing the $yz$-gap again, since that would contradict the fact that we picked a minimal ear. Hence we leave $P$ by crossing the original path of ear. This gives a way to assign, for each newly created intersection with $\ell$, a unique intersection with $\ell$ that was removed, showing that the transformation of the ear does not increase the total number of intersections with $\ell$. In particular the number of self-intersections does not increase since we can choose $\ell$ to be the loop containing the ear in question.

  The second step, removing end-ears, is similar to the first one, except that we have to deal with the endpoint $v$ separately. 
  Let $y$ be the point where an end-ear crosses a gap $a$ incident to $v$ (either $0$ or $1$). Similarly as for ears, by $vy$-gap we mean the points of gap $a$ strictly between $v$ and $y$. An end-ear is minimal, if no other end-ear crosses gap $a$ through the $vy$-gap.  We will remove the end-ears one by one, always picking a minimal end-ear.
  
  Since the end-ear is contained in one of the halves of the sphere, it partitions it into two simply connected sets, one of which, that we denote by $P$, has the $vy$-gap in its boundary.
  We remove the end-ear by continuously transforming it into a path that closely follows the $vy$-gap in the opposite hemisphere, as shown on Fig.~\ref{fig:x-ear_removal}. By choosing the new path sufficiently close to the equator, we can make sure that if a new \mbox{(self-)}intersection with some loop $\ell$ appears, then by tracing $\ell$ from that \mbox{(self-)}intersection in a certain direction we cross the $vy$-gap, thus entering set $P$.
  \begin{figure}[h]
    \centering
    \includegraphics[width=.85\textwidth]{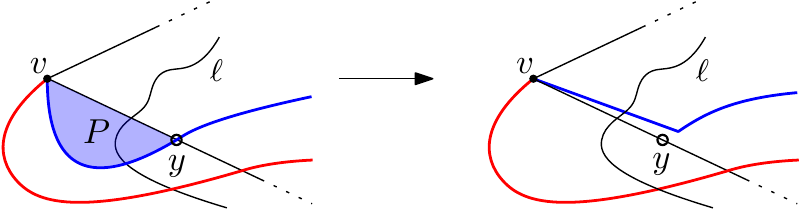}
    \caption{Removal of a minimal end-ear.}%
    \label{fig:x-ear_removal}
  \end{figure}
  Tracing $\ell$ further we must eventually leave the set $P$, since $v \notin P$.
  This cannot happen by crossing the $vy$-gap again, since that would contradict the fact that we removed all ears in the first step. It also cannot happen by crossing~$v$, since this would contradict that we chose a minimal end-ear. Hence we leave $P$ by crossing the original path of the end-ear, which determines an intersection with the loop $\ell$ that was removed by transforming the end-ear. 
  
  Similarly as in the first step this assigns a unique removed intersection with~$\ell$ to each new intersection with $\ell$, showing that removal of a minimal end-ear does not increase the number of (self-)intersections.

  We recap what we have proved: by repeatedly removing minimal ears in the first step and removing minimal end-ears in the second step we end up with a drawing which does not have any ears nor end-ears, proving the lemma.
\end{proof}

\begin{proof}[Lemma~\ref{lem:reduced_words} \ref{enum:x_not_obst}]
  To prove \ref{enum:x_not_obst}, we need to remove only ears, and we do it precisely the same way as in the proof of Lemma~\ref{lem:reduced_words} \ref{enum:x_is_obst}, see Figure~\ref{fig:ear_removal}. We need to check two things: (a) $x$ does not belong to $\overline P$ (thus making every loop $\ell$ that crosses the new route exit $P$ through the old route); (b) removal of a single ear does not violate the property that there is a path from $x$ to some obstacle $u \in V_n$ that does not intersect any loop or the equator. For (a), note that, $x$ is not contained in the boundary of $P$ because this boundary consists of a segment of the equator and a segment of some loop that is not incident to $x$; neither is $x$ contained in $P$, since otherwise the path which connects $x$ to some obstacle $u$ (and does not intersect any loop or the equator) would be contained in $P$ by connectedness of $P$, implying that $u$ belongs to $\overline P$, which is impossible. For (b), suppose for contradiction, that $\pi$ is the $x$-$u$ path and no matter how close to the $yz$-gap we choose the new route, it intersects $\pi$. This would imply (due to compactness) that some point in $\pi$ belongs to the closure of the $yz$-gap. Since this closure contains neither $u$ nor $x$, it implies that some internal point of $\pi$ belongs to it, thus contradicting that $\pi$ does not intersect the equator. This shows that (b) holds.
\end{proof}
}

\subsection{Characterization of homotopic loops}
Once we can simplify loops using Lemma~\ref{lem:reduced_words}, we can use the following lemma to describe the homotopy classes of loops in terms of induced words. 

Recall that $v$-loops start and end at $v_1$, which is incident to gaps $0$ and $1$.
\begin{lemma}\label{lem:nonhomotopicwords}

  \begin{romenumerate}
    \item\label{en:nonhomotopicwords:xloop}
      Two $x$-loops inducing $\alpha\alpha$-free words $w_1$ and $w_2$ are homotopic if and only if $w_1 = w_2$.
    \item\label{en:nonhomotopicwords:samehemi}
      Suppose that two $v$-loops $\ell_1$ and $\ell_2$ are such that all four initial/final arcs lie in the same hemisphere. Suppose $\ell_1, \ell_2$ induce $\alpha\alpha$-free inner words
      \begin{equation}\label{eq:twowords}
	w_1 = x_1u_1z_1, \quad \text{and} \quad w_2 = x_2u_2z_2,
      \end{equation}
      where (possibly zero-length) words $x_1, x_2, z_1, z_2$ use only letters $0$ and $1$ and words $u_1, u_2$ start and end in letters other than $0$ or $1$ (if $u_i$ is empty, we assume that $x_i = w_i$ and $z_i$ is empty).  
      Then $\ell_1$ and $\ell_2$ are homotopic if and only if $u_1 = u_2$ and the lengths of $x_1$ and $x_2$ have the same parity.
    \item\label{en:nonhomotopicwords:noprefix}
      If two $v$-loops $\ell_1$ and $\ell_2$ induce words $w_1$ and $w_2$ that start and end in a letter other than $0$ or $1$, then $\ell_1$ and $\ell_2$ are homotopic if and only if $\ell_1$ and $\ell_2$ start in the same hemisphere, end in the same hemisphere, and $w_1 = w_2$.
  \end{romenumerate}
\end{lemma}
\ifarxivversion
Lemma~\ref{lem:nonhomotopicwords} is proved in Appendix~\ref{apx:nonhomotopicwords}. 
\fi
It is based on the description of the (fundamental) homotopy group of the plane with several points removed and the correspondence between words and the generators of this group. 

The idea of the part \ref{en:nonhomotopicwords:samehemi} is that we may ``unwind'' the initial and final segments that just ``wind around'' the obstacle $v$, without changing the homotopy class. The reason why, say, the prefixes $x_1$ and $x_2$ have to have matching parities (and merely $u_1 = u_2$ is not enough), is that otherwise the segments corresponding to the subwords $u_1$ and $u_2$ would have opposite polarity. One can see, say, that if two loops start and end in the northern hemisphere and induce words $02$ and $20$ (so that $u_1 = u_2 = 2$), they are not homotopic (in fact, $\ell_1$ is homotopic to $\ell_2$ reversed).
\toappendix{
\subsection{Proof of Lemma~\ref{lem:nonhomotopicwords}}\label{apx:nonhomotopicwords}
\begin{proof}[Lemma~\ref{lem:nonhomotopicwords}]
  In the proof we treat $S$ as a plane with $n$ points removed.

  We first consider the case of $x$-loops, \ref{en:nonhomotopicwords:xloop}. For every $i \in [n]$, let $g_i$ be an $x$-loop without self-intersections that circles obstacle $v_i$ so that the gap $i-1$ is crossed before the gap $i$ (thus $v_i$ is the only obstacle contained in the bounded set surrounded by $g_i$. By reversing the orientation of $g_i$, we obtain the \emph{inverse} of $g_i$, denoted by $g_i^{-1}$. It is a well known from algebraic topology  that every $x$-loop $\ell$ in $S = \R^2 \setminus V_n$ is homotopic to a concatenation of a finite sequence of \emph{elementary loops} $g_1, \dots, g_n, g_1^{-1}, \dots, g_n^{-1}$. In each homotopy class, there is a unique concatenation of elementary loops in which no two consecutive elementary loops are inverses of each other (we call such concatenations \emph{reduced}). In other words, homotopy classes of $x$-loops under the operation of concatenation form a free group with generators $g_1, \dots, g_n$. 
    
    We describe a correspondence between the word induced by $\ell$ (note that it necessarily has even length) and its homotopy class written as concatenation of elementary loops. Elementary loop $g_i$ induces a two-letter word $(i-1) i$ and $g_i^{-1}$ induces a word $i(i - 1)$. 
    For $0 \le i < j \le n$, define $g_{i,j}:= g_{i+1}g_{i+2}\dots g_j$ and for $0 \le j < i \le n$, define $g_{i, j} := g_i^{-1}g_{i-1}^{-1}\dots g_{j+1}^{-1}$.
    Note that, for any distinct $i,j$, loop $g_{i,j}$ is homotopic to any loop that induces a word $ij$. 
    This gives a bijection between even-length $\alpha\alpha$-free words and reduced concatenations $g_{i_1}^{\eps_1} \dots g_{i_k}^{\eps_k}$, where $\eps_j \in \{-1,1\}$. To be precise, given a word $w = w_1 \dots w_{2m}$, we partition it into words of length two and replace each word $ij$ by $g_{i,j}$. 
    Note that if $w$ is $\alpha\alpha$-free, the resulting concatenation of elementary loops is reduced.

      In particular, if two $x$-loops induce $\alpha\alpha$-free words, their homotopy classes are different if and only if the words are different.

      \smallskip 

      Now we deal with the case of $v$-loops, $v = v_1 \in V_n$, starting with the setting of Lemma~\ref{lem:nonhomotopicwords}.\ref{en:nonhomotopicwords:samehemi}.
      We can characterize when two $v$-loops $\ell_1, \ell_2$ are homotopic by turning them into $x$-loops with $x \notin V_n$ (cf. proof of Theorem~2 in \cite{PTT}).
      Fix a circle $C$ centered at $v$ so that all obstacles except of $v$ lie outside of it, and pick an arbitrary point $x$ in the northern hemisphere that lies on the circle $C$.
      Given a $v$-loop $\ell$, let $\alpha_\ell$ and $\omega_\ell$ be the points where $\ell$ hits the circle the first and the last time, respectively.
      Define a $x$-loop $\ell'$ by concatenating (i) the arc of the circle $C$ between $x$ and $\alpha_\ell$ that does not cross the equator, (ii) the segment of $\ell$ between $\alpha_\ell$ and $\omega_\ell$, and (iii) the arc of the circle $C$ between $\omega_\ell$ and $x$ that does not cross the equator, see Figure~\ref{fig:deobstaclisation}.
      We can assume, for $i = 1, 2$, that $\ell_i'$ induces the same word as the inner word induced by $\ell_i$.
      This is because for $C$ small enough, points $\alpha_{\ell_i}, \omega_{\ell_i}$ lie in the northern hemisphere (by our assumption on the initial and final arcs of the loops) and $\ell_i$ does not cross the equator between $v$ and $\alpha_{\ell_i}$ nor does it between $\omega_{\ell_i}$ and $v$. 
  \begin{figure}[h]
    \centering
    \includegraphics[width=.5\textwidth]{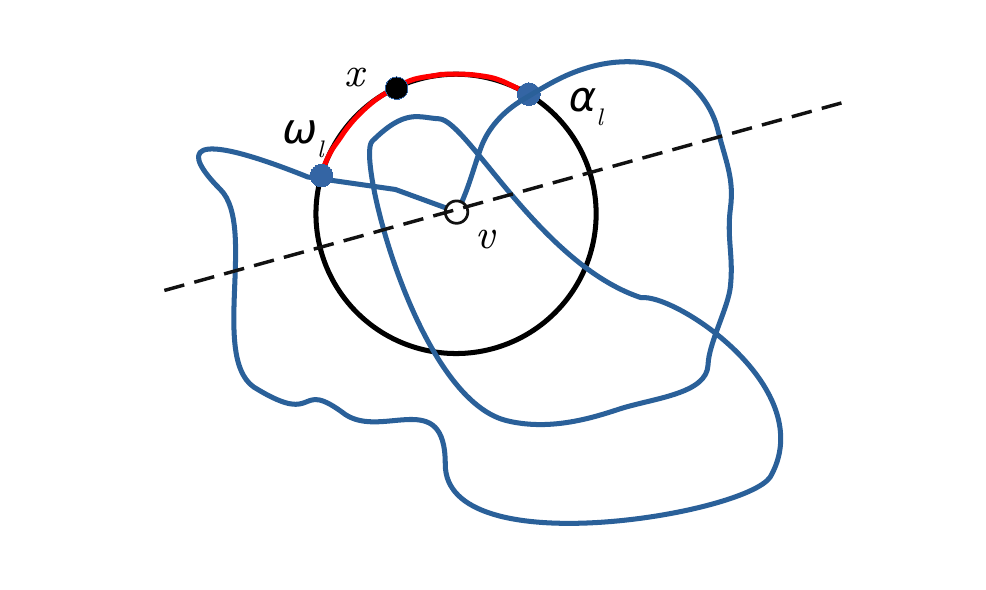}
    \caption{Converting a $v$-loop to an $x$-loop; equator denoted by a dashed line.}%
    \label{fig:deobstaclisation}
  \end{figure}
  
  Assume that an $x$-loop $\ell'$ induces a word $w = w_*\hat w w^*$, where $w_*$ and $w^*$ are maximal words consisting of letters $0$ and $1$. Suppose that $\ell' \sim g_1^mh g_1^{m'}$, where $h$ is a reduced concatenation of elementary loops starting and ending in a loop other than $g_1$ or $g_1^{-1}$. 
  We first simplify $w$ without changing $h$. Note that if we delete first two letters of $w_*$ (i.e., $01$ or $10$), we only increase or decrease $m$ by one, but do not change~$h$. Hence we can assume $w_*$ has at most one letter. Similarly we can assume $w^*$ has at most one letter. Further note that if $w_* = 0$, then $m = 1$ and $h$ starts with $g_2$, so replacing it by $w_* = 1$ only changes $m$ to $0$, without changing $h$. Similarly if $w^* = 0$, we have $m' = -1$ and $h$ ends in $g_2^{-1}$, so setting $w^* = 1$ changes $m'$ to $0$ and preserving $h$. Hence, without changing $h$, we can reduce the word $w$ to one of the forms $\hat{w}$, $1\hat{w}$, $\hat{w}1$ and $1\hat{w}1$, which correspond to $h$ of the form $h'$, $g_2h'$, $h'g_2^{-1}$, and $g_2h'g_2^{-1}$, respectively, where $h'$ stands for a reduced concatenation starting with an elementary loop other than $g_2$ and ending with an elementary loop other than $g_2^{-1}$. This implies that if $x$-loops $\ell_1', \ell_2'$ induce words as in \eqref{eq:twowords} and satisfy
  \begin{equation}
\label{eq:pref_suff}
    \ell_1' \sim g_1^{m_1} h^{(1)} g_1^{m'_1}, \quad \text{and} \quad  \ell_2' \sim g_1^{m_2} h^{(2)} g_1^{m'_2},
  \end{equation}
  with $m_i, m_i' \in \Z$ maximal, then $h^{(1)} = h^{(2)}$ if and only if words $u_1$ and $u_2$ are equal and start at a position of the same parity, that is, length of $x_1$ and $x_2$ have the same parity.

  It remains to show that $v$-loops satisfy $\ell_1 \sim \ell_2$ iff $h^{(1)} = h^{(2)}$.
  We modify each $\ell_i$, $i = 1, 2$, without changing its homotopy class, by adding an excursion from $\alpha_{\ell_i}$ to $x$ and back within the northern half of the circle $C$, and adding a similar excursion from $\omega_{\ell_i}$ to $x$ (see Figure~\ref{fig:deobstaclisation}). Thus we can assume $\ell_i$ contains $\ell'_i$ (by which we mean that the function $\ell'_i$ is a restriction of the function $\ell_i$ to a closed subinterval). Also, without changing the homotopy class of $\ell_i$, we can assume that $\ell'_i$ is a concatenation of elementary loops. Consider the initial segment of $\ell_i$ that ends at $x$ after traversing the $g_1^{m_i}$ part of $\ell'_i$. It is easy to see that  if we replace this $v$-$x$ curve by a straight $vx$ segment we do not change the homotopy class of $\ell_i$. Similarly we can replace the final segment of $\ell_i$ that starts at $x$ and traverses the $g_1^{m'_i}$ part of $\ell'_i$ by a straight segment $xv$. Hence the homotopy class of $\ell_i$ does not depend on $m$ and $m'$. In particular, if $h^{(1)} = h^{(2)}$, then $\ell_1 \sim \ell_2$. 

  For the other direction, assume that $h^{(1)} \neq h^{(2)}$. Assuming, for contradiction, that $\ell_1 \sim \ell_2$, we claim that $\ell_1' \sim g_1^i \ell_2'g_1^{j}$ for some $i, j \in \Z$. Recalling \eqref{eq:pref_suff}, this easily implies that $h^{(1)} = h^{(2)}$, giving a contradiction.
  To see the claim, recall that we can assume that, for $k = 1, 2$, loop $\ell_k'$ is a restriction of $\ell_k$ to some sub-interval of $[0,1]$. For simplicity assume that this interval is $[1/3, 2/3]$. Let $H$ be a homotopy of $v$-loops such that $H(0, \cdot) = \ell_1$, $H(1, \cdot) = \ell_2$.
  Using continuity of $H$ and the fact that $H(s, 0) = v$ for every $s \in [0,1]$, one can show that $s \mapsto H(s, 1/3)$ and $s \mapsto H(s, 2/3)$ are $x$-loops with homotopy class of the form $g_1^i$. 
  We define a continuous function $K : [0,1]^2 \to \R^2 \setminus V_n$ by 
  \[
    K(r,t) = \begin{cases}
      H(3rt, 1/3),   & t \in [0, 1/3]\\
      H(r, t), & t \in [1/3, 2/3] \\
      H(3r(1 - t), 2/3), & t \in [2/3, 1].
    \end{cases}
  \]
  Function $K$ is a homotopy of $x$-loops between $K(0, \cdot) \sim \ell_1'$ and $K(1, \cdot) \sim g_1^i\ell_2'g_1^{j}$, implying $\ell'_1 \sim g_1^i \ell'_2 g_1^j$, as desired. 
  This completes the proof of \ref{en:nonhomotopicwords:samehemi}.

  Finally, assume that the $v$-loops satisfy the conditions of the case \ref{en:nonhomotopicwords:noprefix}.
  By rerouting the equator near $v$ (see Figure~\ref{fig:adjusting_equator}) we can make sure that all loops start and end in, say, the northern hemisphere, and at the same time the word induced of each loop gains a letter $1$ at the beginning (if and only if its initial segment was originally in the southern hemisphere), and gains a letter $1$ at the end (if and only if its final segment was originally in the southern hemisphere).
  \begin{figure}[h]
    \centering
    \includegraphics[width=0.4\textwidth]{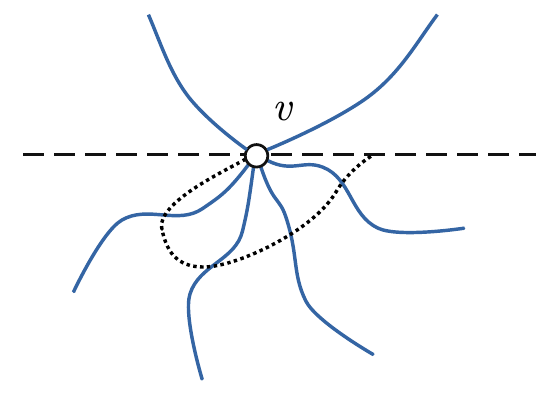}
    \caption{Rerouting of equator (dashed) that makes all loops start and end in the northern hemisphere.}%
    \label{fig:adjusting_equator}
  \end{figure}
  Using the characterization in \ref{en:nonhomotopicwords:samehemi}, this gives the claimed characterization in \ref{en:nonhomotopicwords:noprefix}.
\end{proof}
}


  \subsection{Subwords forcing intersections}
    The equator has two natural orientations. 
    For any ordered triple $(a, b, c)$ of distinct gaps we assign the orientation of the equator such that if we circle the equator starting from the gap $a$, we encounter the gap $b$ before $c$. 
  In particular, for any three distinct gaps $a,b,c$ triples $(a,b,c), (b,c,a), (c,a,b)$ have the same orientation while triples $(a,b,c)$ and $(c,b,a)$ have opposite orientations.
  For example, if $n \ge 2$, recalling our labeling of gaps from Section~\ref{sec:setup} (in particular that the vertex $v = v_1$ is a special gap) we have that $(0,v,1)$, $(v,1,2)$ and $(0,1,2)$ have the same orientation.

  \begin{lemma}\label{lem:crossing_orient}
    Let $k \ge 0$ be an integer. Consider two segments of the same polarity corresponding to $\alpha\alpha$-free words $a_0 a_1\dots a_k a_{k+1}$ and $b_0 b_1 \dots b_k b_{k+1}$ such that for $i = 1, \dots, k$ we have $a_i = b_i$, while $a_0 \neq b_0$ and $a_{k+1} \ne b_{k+1}$. Suppose that $(a_0, b_0, a_1)$ and $(a_{k}, b_{k+1}, a_{k+1})$ have opposite orientations for even $k$, and the same orientation for odd $k$. Then there is $i \in \left\{ 0, \dots, k \right\}$ such that the $a_i a_{i+1}$-arc of the first word intersects the $b_i b_{i+1}$-arc of the second word. 
  \end{lemma}
  \paragraph{Proof idea.}
  Assume, for a contradiction, that the $a_i a_{i+1}$-arc and the $b_i b_{i+1}$-arc are disjoint for every $i$.
  It follows inductively that the orientation of the triple $(a_i, b_{i+1}, a_{i+1})$ is uniquely determined for every $i$.
  In particular, it must be the same as the orientation of $(a_0, b_0, a_1)$ for even $i$ and opposite otherwise. For $i = k$, we arrive at a contradiction.
\ifarxivversion
  The full proof is postponed to Appendix~\ref{apx:crossing_orient}.
\fi
  \toappendix{
  \subsection{Proof of Lemma~\ref{lem:crossing_orient}}\label{apx:crossing_orient}
  \begin{proof}[Lemma~\ref{lem:crossing_orient}]
    For $k \ge 1$, we denote by $x_i$, $i = 0, \dots, k + 1$, the points where the first segment intersects the equator (thus $x_i$ belonging to the gap $a_i$) and by $y_i$, $i = 0, \dots, k + 1$, the points where the second segment intersects the equator (so $y_i$ belonging to the gap $b_i$). We can assume that $x_i \neq y_i$ for every $i$.

    We define the orientation of a triple $(z_1, z_2, z_3)$ of distinct points on the equator similarly as for gaps: if you travel from $z_1$ in this orientation, you encounter $z_2$ before $z_3$. 

    We prove the lemma by contradiction. For this we assume that the corresponding arcs do not intersect and apply the following claim, noting that the orientation of $(a_0, b_0, a_1)$ equals the orientation of $(x_0, y_0, x_1)$ and the orientation of $(a_{k}, b_{k+1}, a_{k+1})$ equals the orientation of $(x_{k}, y_{k+1}, x_{k+1})$.
    \begin{claim}
      Consider two segments of the same polarity, which intersect the equator at disjoint sets of points $(x_i : i = 0, \dots, k+1)$ and $(y_i : i = 1, \dots, k+1)$, respectively. Assume that for every $i = 1, \dots, k$ points $x_i, y_i$ belong to the same gap, and for every $i = 1, \dots, k+1$ points $x_i, x_{i-1}$ belong to different gaps (note that we allow the corresponding terminal points of the segments to belong to the same gap). If, for $i = 1, \dots, k+1$, the arc between $x_{i-1}$ and $x_i$ does not intersect the arc between $y_{i-1}$ and $y_i$, then the orientations of $(x_0,y_0,x_1)$ and $(x_k,y_{k+1}, x_{k+1})$ is the same for even $k$ and opposite for odd $k$. 
    \end{claim}

    We prove the above claim by induction on $k$. The base case $k = 0$ is trivial, since all it says is that if two arcs do not intersect, then travelling along the equator in a certain direction between the endpoints of one of the arcs we will pass through both endpoints of the other arc.

    Now assume $k \ge 1$. By induction hypothesis applied to points $x_0, x_1$ and $y_0, y_1$ we obtain that $(x_0, y_0, x_1)$ has the same orientation as $(x_0, y_1, x_1)$ and by the induction hypothesis applied to points $x_1, \dots, x_{k+1}$ and $y_1, \dots, y_{k+1}$ we get that
    the orientations of $(x_1,y_1,x_2)$ and $(x_k,y_{k+1}, x_{k+1})$ are the same for even $k - 1$ (thus odd $k$) and opposite for odd $k - 1$ (thus even $k$).
    It therefore remains only to verify that 
    \begin{equation}
      \label{eq:step}
      (x_0,y_1,x_1) \text{ and }  (x_1, y_1, x_2) \text{ have opposite orientations.}
    \end{equation}
    To see \eqref{eq:step}, let $a$ be the gap to which points $x_1, y_1$ belong. Removing points $x_1, y_1$ from the equator we obtain two sets $A$ and $B$. One of them, say $A$, must be contained in the gap $a$. Since $x_0$ and $x_2$ do not belong to the gap $a$, they are both in the set $B$ which means $(x_1, x_0, y_1)$ and $(x_1, x_2, y_1)$ have the same orientation which is equivalent to \eqref{eq:step}.
  \end{proof}
  }
  
  \subsection{Windings in $v$-loops}
  We now focus on $v$-loops and describe the intersections forced by specific alternating words.
  
  For a word $w$ we write $w^k$ a concatenation of $k$ copies of $w$, say $(ab)^2 = abab$.
  Given an obstacle $v_i$, $i \ne 1$, let $a, b$ be the gaps incident to $v_i$. For integer $s \ge 1$, an \emph{$s$-winding around $v_i$} is a $w$-segment, where $w$ has a form $tw'u$, where $w'$ is of the form $(ab)^sa$, $(ba)^sb$, $(ab)^{s+1}$ or $(ba)^{s+1}$ and $t, u$ are letters other than $\{a, b\}$.
  Assuming the gaps incident to $v_1$ are $0$ and $1$, an \emph{$s$-winding around $v := v_1$} is a $w$-segment with $w$ of the same form as above (for $\{a,b\} = \{0, 1\}$), but with $t, u$ being letters other than $0, 1$, and $v$ (the difference from the first case is that we do not allow $t= v$ or $u = v$), see the left part of Figure~\ref{fig:winding_snail}.

  The proofs of the following lemmas are sketched at the end of this subsection; the detailed proofs are \ifarxivversion deferred to the appendix \fi.
  
  \begin{figure}[h]
    \centering
    \includegraphics[width=0.7\textwidth]{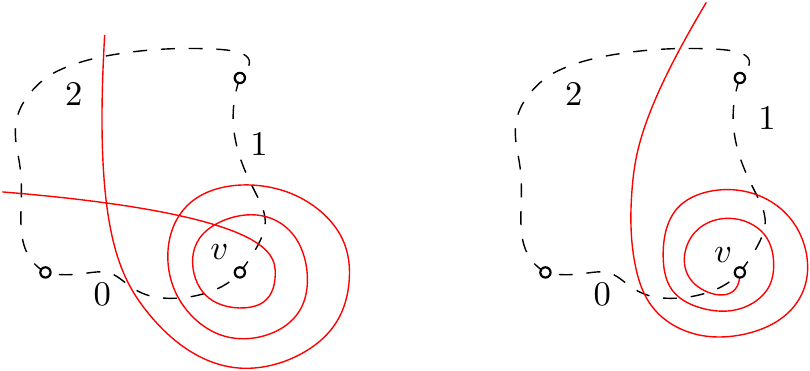}
    \caption{Left: a 2-winding around $v$ with word $2(01)^32$; right: a $(2,2)$-snail with word $2(01)^20v$.}%
    \label{fig:winding_snail}
  \end{figure}
  
\begin{lemma}\label{lem:windings_orient}
Suppose $S$ and $T$ are an $s$-winding and a $t$-winding, respectively, both around the same obstacle $v \in V_n$. Then 
\begin{romenumerate}
\item \label{en:windings_orient:self}
  Segment $S$ has at least $s$ self-intersections.

\item \label{en:windings_orient:inters}
  $S$ and $T$ have at least $2\cdot\min\{s,t\}$ mutual intersections (provided $S\neq T$). 
\end{romenumerate}
\end{lemma}
\toappendix{
\subsection{Proofs of lemmas on the windings in $v$-loops}\label{apx:windings}
\begin{proof}[Lemma~\ref{lem:windings_orient}]
  We make a convention that $c_i$, $i = 1, 2,\dots$ stand for an element of $(V_n \setminus \{a,b\}) \cup \{x\}$ if $v \neq x$ and element of $V_n \setminus \{a,b\}$ if $v = x$.
  We say that an $s$-winding is \emph{even} if it is of the form $c_1(ab)^{s+1}c_2$ and otherwise it is \emph{odd} and has form $c_1(ab)^s ac_2$. Note that by reversing an odd winding we change its polarity without changing the order of the letters in the parentheses. On the other hand, by reversing an even winding, we change the order of letters in the parentheses without changing polarity. 

  Given a pair of \emph{even} windings, we can assume without loss of generality (by optionally renaming the letters $a,b$ and/or reversing), that it is of one of the following forms:
  \begin{enumerate}
    \item \label{en:EEsymm}
      $c_1(ab)^{s+1}c_2$, $c_3(ab)^{t+1}c_4$, same polarity
    \item \label{en:EEasym}
    $c_1(ab)^{s+1}c_2$, $c_3(ba)^{t+1}c_4$, opposite polarity
  \end{enumerate}		
  Given pair of an \emph{even} winding and an \emph{odd} winding we can write it in the form
  \begin{enumerate}[resume]
    \item \label{en:EO}
    $c_1(ab)^{s+1}c_2$, $c_3(ab)^{t}ac_4$, same polarity
  \end{enumerate}	
  Finally, given a pair of \emph{odd} windings, we can write it in one of the following forms: 
  \begin{enumerate}[resume]
    \item \label{en:OOsym}
    $c_1(ab)^{s}ac_2$, $c_3(ab)^{t}ac_4$, same polarity
    \item \label{en:OOasym}
    $c_1(ab)^{s}ac_2$, $c_3(ba)^{t}bc_4$, opposite polarity
  \end{enumerate}		
  We claim that for each of the above forms there are at least $\min \{s,t\}$ pairs of segments (with the same number of letters) of the same polarity consisting of an initial segment of the $s$-winding and a final segment of the $t$-winding that intersect. Namely, in each case we choose these pairs so that they match one of the forms: (A) $(c_1(ab)^ia, b(ab)^ic_4)$ or (B) $(c_1(ab)^i, (ba)^ic_4)$.

  Lemma~\ref{lem:crossing_orient} implies that a segment pair (of the same polarity) of form (A) or (B) implies an intersection between corresponding arcs (for example, if we choose initial-final segments of the form (A), for $i = 1, 2, \dots$, then for every $i$ there is $j = j(i)$ so that an intersection occurs between the arcs corresponding to $j$th and $(j+1)$th letter in each of the words $c_1(ab)^ia$ and $b(ab)^ic_4$).
This makes sure that each initial-final segment pair gives entails a distinct intersection between the $s$-winding and the $t$-winding.

There is one thing that needs some care. Lemma~\ref{lem:crossing_orient} assumes the same polarity of the segments. Note that all initial segments of the $s$-winding inherit its polarity. On the other hand, the polarity of the final segment of the $t$-winding alternates as the number of letters increases. One can check, in each of the five cases, that the segments that we chose indeed have the same polarity. For example, in the case \ref{en:EEsymm}, we have to choose even-letter segments, since they are forced by the form of the $t$-winding. But since the windings have the same polarity, the $t$-winding is even, and we choose even-length segments, the final segment of the $t$-winding turns out to have the same polarity as the $s$-winding. The remaining cases can be checked similarly. 

It is also easy to see, that in each case the number of pairs we chose is at least $\min \{s,t\}$. In particular this implies part \ref{en:windings_orient:self} of the lemma, since the argument applies to the case $S = T$.

Similarly, we count number of pairs of final segments of the $s$-winding and initial segments of the $t$-winding that force intersections. The form of the segments now is not precisely of the form (A) or (B), but may require swapping the letters $a$ and $b$ (which clearly does not change the condition required by Lemma~\ref{lem:crossing_orient}). Namely, we choose even-length segments in the cases \ref{en:EEsymm}, \ref{en:EO}, and \ref{en:OOasym}, while we choose odd-length segments in the cases \ref{en:EEasym} and \ref{en:OOsym}. Again the number of choices is at least $\min\left\{ s, t \right\}$ in each of the five cases. It is easy to see that for $S \neq T$ the intersections implied by final-initial segment pairs are different from those implied by the initial-final segment pairs we discussed before. This completes the proof of \ref{en:windings_orient:inters}.
\end{proof}
}

We excluded from the definition of windings the case where the obstacle is $v_1$ and the alternating sequence appears at the beginning/end of the (inner) word.	
Given a positive integer $m$, let $(ab)^{-m} := (ba)^m$ and let $(ab)^0$ stand for an empty word. Given an integer $s$ and a letter $a$ other than $0, 1$ (but possibly $v$), by a \emph{$(s,a)$-snail} we call a $w$-segment where $w$ is a $\alpha\alpha$-free word of the form 
\[
  v(01)^sw'a, \text{ where } w' \in \{0, 1, 01, 10\}, a \notin \{0,1\},
\]
see the right part of Figure~\ref{fig:winding_snail}.
(Note: since $w$ is $\alpha\alpha$-free, we cannot, say, have $w' \in \{1, 10\}$ if $s > 0$.)

  \begin{lemma}\label{lem:x_windings}
    Consider a $(s,a)$-snail and a $(t,b)$-snail of the same polarity. 
    If $s t < 0$, the snails intersect at least $\min \{|s|, |t|\}$ times. 
    If $s t > 0$ and $a, b \neq v$, then the snails intersect at least $|s - t| - 1$ times. 
%
%
%
  \end{lemma}

\toappendix{
  \begin{proof}[Lemma~\ref{lem:x_windings}]
    Assume that the $(s,a)$-snail induces word $w^{(1)}$ and the $(t,b)$-snail induces word~$w^{(2)}$.

    Case $st < 0$. Assume without loss of generality that $t < 0 < s$, so that $w^{(1)}$ starts with $\tl(01)^{s}0$, and  $w^{(2)}$ starts with $\tl(10)^{-t}1$ with $-t > 0$.
    For $i = 1, \dots, \min \{s,-t\}$ consider initial segments of $w^{(1)}, w^{(2)}$ of length $2i + 2$ (the second one reversed), so that they induce words
    \begin{equation}
      \label{eq:opposite_tuples_shorter}
      \tl(01)^i0, \text{ and } 1(01)^i\tl,
    \end{equation}
    (Note that reversing of the second segment did not change its polarity, due to even length.)
    Triples $(\tl,1,0)$ and $(1,\tl,0)$ have opposite orientations.
    Together with the fact that the initial segment of $w^{(2)}$ is reversed, via Lemma~\ref{lem:crossing_orient} we get that pairs \eqref{eq:opposite_tuples_shorter} force distinct intersections, giving $\min \{|s|, |t|\}$ intersections, as claimed.

    \medskip
    Case $st > 0$. Assume without loss of generality that $s \ge t > 0$. We can further assume that  $s \ge t + 2$, since otherwise the bound $|s-t| - 1$ we are aiming at is trivial. Note that $w^{(1)}$ starts with $\tl(01)^{s}0$.
    We consider two possible cases: (i) $w^{(2)} = \tl(01)^{t+1}b$, $b \notin \{\tl, 0, 1\}$ and (ii) $w^{(2)} = \tl(01)^t0b$, $b \notin \{0, \tl,1\}$.
    In the case (i) for $i = 1, \dots, s - t - 1$ we can write
\[
  w^{(1)} = \tl0(10)^{i-1}\underline{1(01)^{t+1}0}(10)^{s-t-i-1}\dots 
\]
and comparing the underlined word with the word $w^{(2)} = v(01)^{t+1}b$ of the $(t,b)$-tail, both of which have the same, even, length.
Since $ b \notin \{0,\tl,1\}$, triples $(1,\tl,0)$ and $(1,b, 0)$ have opposite orientations. So by Lemma~\ref{lem:crossing_orient} the underlined word and $w^{(2)}$ force an intersection which is different for each $i = 1, \dots, s-t-1$.

    In the case (ii), writing, for $i = 1, \dots, s- t - 1$,

\[
  w^{(1)} = \tl0(10)^{i-1}\underline{1(01)^t01}(01)^{s-t-i-1}\dots
\]
we compare the underlined word and $w^{(2)} = v(01)^t0b$, which have the same, odd, length.
Since $b \notin \{0,\tl,1\}$, triples $(1,\tl,0)$ and $(0,b, 1)$ have the same orientation, and so by Lemma~\ref{lem:crossing_orient} the underlined word and $w^{(2)}$ force an intersection and for each $i = 1, \dots, s-t-1$ the intersection is different.
  \end{proof}
  }
  The following lemma is used in the proof of the inequality \eqref{eq:f_less_g} of Proposition~\ref{prop:ineqs}.

  \begin{lemma}\label{lem:snail_pandemonium}
Fix a word $u$ that starts an and ends in a letter other than $0$ or $1$. If we have a family $F$ with more than $4(2k + 1)^2$ $v$-loops of the same polarity, each of which induces an even-length $\alpha\alpha$-free word of the form
  \begin{equation}
    \label{eq:word_form}
    \tl(01)^{s}w'uw''(10)^{t}\tl, \quad t, s \in \Z, \quad w', w'' \in \{0, 1, 01, 10\},
  \end{equation}
then there are $\ell_1, \ell_2 \in F$ (possibly $\ell_1 = \ell_2$) with at least $k$ (self)-intersections.
  \end{lemma}
  
  \toappendix{
  \begin{proof}[Lemma~\ref{lem:snail_pandemonium}]
    If $u$ is an empty word, then each loop induces a word that is of the form $\tl(01)^s01\tl$ and therefore is a $(s,\tl)$-snail. Since reversing an even-length segment does not change its polarity, the same loop reversed is a $(-s,\tl)$-snail of the same polarity, and thus has $|s|$ self-intersections by Lemma~\ref{lem:x_windings}. If $|F| \ge 2k$, then $F$ contains a loop that induces $\tl(01)^s01\tl$ with $|s| \ge k$ which has at least $k$ self-intersections.

    Assuming $u$ is nonempty, let a \emph{$s$-snail} mean a $(s,a)$-snail for any $a \neq \tl$. So a loop that induces a word of the form \eqref{eq:word_form}, starts with a $s$-snail and ends with a $t$-snail, to such loop we associate a pair $(s, t)$. By the assumptions on the loops and words, all these snails have the same polarity. 

    Let $S$ be the set of integers $s$, such that some $\ell \in F$ contains a $s$-snail.
    In view of Lemma~\ref{lem:x_windings}, it is enough to show that $S$ contains two numbers $s, t$ such that either (i) $st < 0$ and $\min \{s, t\} \ge k$ or (ii) $st > 0$ and $|s - t| \ge k + 1$.
    
      For given $s, t$, there are at most $4$ words of the form \eqref{eq:word_form}, hence if the family has more than $4(2k + 1)^2$ loops, there are more than $(2k + 1)^2$ different pairs $(s, t)$ that are associated to some $\ell \in F$, in particular $|S|^2 > (2k + 1)^2$. Denoting $S_+ := \{s \in S : s > 0\}$, $S_- := \{s \in S : s < 0\}$, we have that $|S_+| + |S_-| \ge 2k + 1$. If $|S_+|, |S_-| \ge k$, there are numbers $s,t$ satisfying (i). Otherwise, $\min \{|S_+|, |S_-| \} \le k - 1$ and so $\max \{|S_+|, |S_-| \} \ge k + 2$, in which case there are $s, t$ satisfying (ii).
  \end{proof}
  }

   We note in passing that Lemmas~\ref{lem:windings_orient} and~\ref{lem:x_windings} are proved by finding sufficiently many segment pairs that satisfy the conditions of Lemma~\ref{lem:crossing_orient} and arguing that each pair implies a distinct intersection. Lemma~\ref{lem:snail_pandemonium} then follows from Lemma~\ref{lem:x_windings} by relatively straightforward pigeonhole-type arguments.

\ifarxivversion
  The full proofs of Lemmas~\ref{lem:windings_orient} to \ref{lem:snail_pandemonium} are included in Appendix~\ref{apx:windings}.
\fi

  \section{Expansions of words}
  \label{ss:expansions}

Given an $\alpha\alpha$-free word $w$ of length at least two, consider all maximal subwords that use two letters. For example, if $w = 2010212$, the maximal words are $20,010,02$, and $212$. Ordering these words by the position of the first letter, it is clear that every two consecutive words overlap in a single letter. We classify these subwords according to the pair of letters they use. For distinct $a, b$ a maximal subword that uses $a$ and $b$ is called an \emph{$ab$-word} (the ordering of $a$ and $b$ does not matter). For a given pair $a, b$, the $ab$-words are disjoint and surrounded by letters other than $a$ and $b$ (we assume that $w$ is extended by adding the letter $\tl$ at each end).

Each $ab$-word starts either with $ab$ or $ba$. If we replace each of these two-letter subwords by its power, say, $ab$ by $(ab)^{s + 1}, s \ge 0$, we obtain another $\alpha\alpha$-free word, which we call an $ab$-expansion of $w$.
Hence if $w$ has $\ell$ $ab$-words, each $ab$-expansion of $w$ is uniquely described by a vector $s = (s_1, ..., s_{\ell})$ of nonnegative integers. 

If, in addition, $w$ is $\alpha\beta\alpha\beta$-free, then for each distinct $a,b$ all maximal $ab$-words have at most three letters. If $w$ is not $\alpha\beta\alpha\beta$-free, it can be obtained by consecutive $ab$-expansions of a $\alpha\beta\alpha\beta$-free word, one for each pair $a,b$ of distinct letters.



The \emph{self-intersection number} of a loop-word $w$ is defined as the smallest number of self-intersections in a loop that induces $w$.
\begin{lemma}\label{lem:expansions}
  Let $a, b$ be two gaps adjacent to the same obstacle $v$. Let $\ell = \ell(k) \ge 2\sqrt{k}$ be a positive integer.

  Let $w$ be a $\alpha\alpha$-free word which contains no subword $abab$ or $baba$. Suppose that there are at most $\ell$ maximal $ab$-words in $w$.  If $v = v_1$, then also assume that none of these words appears at the beginning or the end of $w$. 
  The number of $ab$-expansions of $w$ with the self-intersection number smaller than $k$ is
  \begin{equation*}
   \left(\frac{\ell}{\sqrt{k}}\right)^{O\left(\sqrt{k} \right)}
      e^{ O\left(\sqrt{k} \right)}.
\end{equation*}
In particular, if also $\ell=O\big(\sqrt{k}\big)$ then the above estimate becomes $e^{ O\left(\sqrt{k} \right)}$.
\end{lemma}

\subsection{Sketch of the proof of Lemma~\ref{lem:expansions}}
    
  Let $\ell'$ be the number of maximal $ab$-words in $w$ and consider an $ab$-expansion of $w$ determined by vector $s_1, \dots, s_{\ell'}$. In this expansion the $i$th maximal $ab$-word, together with the letters surrounding it, is either an $s_i$-winding or an $(s_i + 1)$-winding around $v$. By Lemma~\ref{lem:windings_orient} such an $ab$-expansion has at least
    \begin{equation}
    \label{eq:self_int_upper}
        \sum_i {s_i} + 2 \sum_{i < i'} \min\{s_i, s_{i'}\}
    \end{equation}
    self-intersections. 
    
    Our goal is to give an upper bound, in terms of $k$ and $\ell$, on the number of vectors $s$ of length $\ell'$ such that \eqref{eq:self_int_upper} is smaller than $k$. Since $\ell' \le \ell$ and the number of such vectors is clearly largest for $\ell' = \ell$, let us further assume that $\ell' = \ell$.

%
    
   
    For $i=0,\dots,k$, let $m_i=m_i(s)$ denote the multiplicity of $i$ in $s$,
 formally
 \begin{equation*}
   m_i=m_i(s) := |\left\{ j \in [\ell] : s_j = i \right\}|.
 \end{equation*}
 Given an integer $\alpha \ge 0$, let us write $m_{\ge \alpha} := \sum_{i \ge \alpha} m_i$ and note that
 \begin{equation}
   \label{eq:m_sum}
   m_{\ge 0} = \ell.
 \end{equation}
 Moreover 
 \begin{equation}
   \label{eq:m_tail}
   m_{\ge \alpha} \le \sqrt{k/\alpha}, \qquad \alpha = 1, 2, \dots,k,
 \end{equation}
 since otherwise, noting that $m_{\ge \alpha} = \left\{ i : s_i \ge \alpha \right\}$, by \eqref{eq:self_int_upper} the self-intersection number is at least 
 \[
   m_{\ge \alpha} \alpha + 2 \binom{m_{\ge \alpha}}{2} \alpha = m_{\ge \alpha}^2 \alpha > k,
 \]
giving a contradiction.
 
 An upper bound on the number of vectors $s$ can be obtained by bounding (i) the number of vectors $s$ giving the same vector $m=(m_0,\dots,m_k)$ and (ii) the number of distinct vectors $m=(m_0,\dots,m_k)$ with nonnegative integer coordinates satisfying the constraints~\eqref{eq:m_sum} and \eqref{eq:m_tail}. The product of the two obtained bounds is then an upper bound on the number of vectors $s$.
 The two bounds are obtained by combinatorial methods in \ifarxivversion Appendix~\ref{apx:proof-expansions} \fi, and they are
$({\ell}/{\sqrt{k}})^{\sqrt{k}}
      e^{ O\left(\sqrt{k} \right)}$
for (i) and
$e^{O(k^{1/3}(\ln \ell + \ln k))}$
for (ii).
Their product is of order $({\ell}/{\sqrt{k}})^{O\left(\sqrt{k}\right)}
      e^{ O\left(\sqrt{k} \right)}$, which gives Lemma~\ref{lem:expansions}.

\toappendix{
\section{Estimates needed in the proof of Lemma~\ref{lem:expansions}}\label{apx:proof-expansions}

Here we complete the proof of Lemma~\ref{lem:expansions} by showing the two bounds claimed at the end of Section~\ref{ss:expansions}.

  \subsection{Number of vectors $s$ corresponding to a particular vector $m$}
Our first task is to find an upper bound on the number of choices of $s = (s_1, \dots, s_\ell)$ corresponding to a fixed vector $m$.
By \eqref{eq:m_tail}, we have $m_i = 0$ for $i > k$, so we further treat $m$ as a vector $(m_0, \dots, m_k)$. Note that $m_{\ge 0} = \ell$.
Thus, we are counting the ways to choose, for $i = 0, \dots, k$, which $m_i$ coordinates of $s$ are assigned value $i$, which is the \emph{multinomial coefficient}
    \begin{equation*}
    \binom{\ell}{m_0, \dots, m_{k}} = \frac{\ell!}{\prod_{i = 0}^k m_i!}.
    \end{equation*}

   
    \begin{proposition}\label{prop:multinomial}
    Suppose that 
    \begin{equation}
      \label{eq:ell_cond}
      \ell   \ge  2\left\lfloor\sqrt{k}\right\rfloor
      - \left\lfloor\sqrt{k/2}\right\rfloor.
    \end{equation}
    If a vector $m=(m_0,\dots,m_k)$ with nonnegative integer coordinates satisfies~\eqref{eq:m_sum} and~\eqref{eq:m_tail}, then 
         \begin{equation*}
    \binom{\ell}{m_0, \dots, m_{k}} \le \binom{\ell}{z_0, \dots, z_{k}},
    \end{equation*}
      where
         \begin{align*}
     z_0 & :=  \ell-\left\lfloor\sqrt{k}\right\rfloor,\\
     z_i & :=  \left\lfloor\sqrt{ k/i}\right\rfloor-\left\lfloor\sqrt{ k/(i+1)}\right\rfloor,\ \ \ \ \ \  \text{for $i=1,\dots,k$.}
         \end{align*}
    \end{proposition}

 \begin{proof}
   First observe that we can restrict $m$ to non-increasing vectors, because (a) the multinomial coefficient and the sum of coordinates does not change the value if the coordinates are rearranged, and (b) given $m$, its non-increasing rearrangement $m' = (m'_0, \dots, m'_k)$ satisfies $\sum_{i \ge \alpha} m'_i \le \sum_{i \ge \alpha} m_i = m_{\ge \alpha}$ for $\alpha = 1, \dots, k$, (``the sum of \emph{smallest} $k - \alpha + 1$ coordinates is at most the sum of \emph{some} $k - \alpha + 1$ coordinates''). 

   For a non-increasing vector $m$, the proposition is proved by transforming such $m$ into the vector $z=(z_0,\dots,z_k)$ in finitely many \emph{steps} (defined below).
A variable vector $x=(x_0,\dots,x_k)$ is set to $x=m=(m_0,\dots,m_k)$ at the beginning of the transformation process, and the process terminates with $x=z=(z_0,\dots,z_k)$.
In each of the steps the value of the multinomial coefficient $\binom{\ell}{x_0, \dots, x_{k}}$ does not decrease, which gives the required inequality.


   
   

As with the notation $m_{\ge \alpha}$, for $\alpha\in\{0,\dots,k\}$, we define $x_{\ge\alpha} = x_\alpha+\dots+x_k$ and $z_{\ge\alpha} = z_\alpha + \dots + z_k$.

 
\emph{
The following system $\cal S$ of $k+1$ constraints is satisfied by the vector $x=(x_0,\dots,x_k)$ during the whole transformation process:
\begin{align}
    x_{\ge 0} & =  \ell
    & (\ &= z_{\ge 0\ } ), \text{ and} \label{e:system2} \\
    x_{\ge\alpha} & \le \left\lfloor\sqrt{ k/\alpha}\right\rfloor\ \ \ 
    & (\ &= z_{\ge\alpha}\  ),\ \ 
    \text{for $\alpha=1,\dots,k$}. \label{e:system1}
 \end{align}
 }
 By \eqref{eq:m_sum} and \eqref{eq:m_tail}, the system $\cal S$ is also satisfied by the (initial) vector $x=m$.
   

   
   A step consists of choosing indices $a < b$, decreasing the value of $x_a$ by one and increasing the value of $x_b$ by one. The indices are chosen as follows.
   Suppose that $x=(x_0,\dots,x_k)\ne z$ is a vector satisfying
   the system $\cal S$.
%
     Consider the largest index $b$ such that $x_b\ne z_b$.
     By \eqref{e:system1}, we have $x_{\ge b} \le z_{\ge b}$. Since $x_i = z_i$ for $i > b$, this implies $x_b < z_b$.
     Since $x_{\ge 0}=z_{\ge 0}\  (\ =\ell\ )$, there is also an index $i < b$ such that $x_i > z_i$. Let $a$ be the largest among such indices $i$. This choice of $a$ and $b$ ensures that 
     \begin{equation}
       \label{eq:x_more_z}
       x_b \le z_b - 1 \text{ and }x_j \le z_j \text{ for } j = a +1, \dots, k.
     \end{equation}
   
  
%
     We now verify that the system $\cal S$ stays valid after each step.
     It is immediate that the constraint \eqref{e:system2} is preserved.
     Regarding \eqref{e:system1}, the sum $x_{\ge\alpha}$ increases only if $a < \alpha \le b$. Recalling that \eqref{eq:x_more_z} was true before the step, we conclude that 
     $x_{\ge\alpha}\le z_{\ge\alpha}$ \emph{after} the step, which verifies that the inequality~\eqref{e:system1}
     is preserved. 
   
   Preparing to show that the multinomial coefficient does not decrease during each step, we claim that for $0 \le a < b \le k$ 
   \begin{equation}
     \label{eq:almost_decreasing}
     z_a + 1 \ge z_b.
   \end{equation}
   To see \eqref{eq:almost_decreasing}, consider first the case $a \ge 1$ and note that function $x \mapsto \sqrt{k/x} - \sqrt{k/(x+1)}$ is strictly decreasing on $(0, \infty)$ (this can be checked, say, by calculating the derivative), so
  \begin{align*}
     z_a +1 &= \left\lfloor\sqrt{ k/a}\right\rfloor-\left\lfloor\sqrt{ k/(a+1)}\right\rfloor +1
     \ge   \left\lfloor\sqrt{ k/a}\right\rfloor-\sqrt{ k/(a+1)} +1 \\
     &> \sqrt{ k/a} - \sqrt{ k/(a+1)} > \sqrt{ k/b} - \sqrt{ k/(b+1)} \\
     &> \left\lfloor\sqrt{ k/b}\right\rfloor - \left(\left\lfloor\sqrt{ k/(b+1)}\right\rfloor +1\right) =z_b-1, 
   \end{align*}
   which implies \eqref{eq:almost_decreasing}, since $z_a, z_b$ are integers. In the remaining case $a = 0$, by the assumption \eqref{eq:ell_cond} on $\ell$, and \eqref{eq:almost_decreasing} applied for $a = 1$,
   \begin{align*}
     z_a+1 &= z_0 +1 = \ell-\left\lfloor\sqrt{k}\right\rfloor +1
     \ge  2\left\lfloor\sqrt{k}\right\rfloor 
     - \left\lfloor\sqrt{k/2}\right\rfloor
     -\left\lfloor\sqrt{k}\right\rfloor +1
     \\
     &= z_1 + 1 \ge z_b.
   \end{align*}

   The value of $\binom{\ell}{x_0, \dots, x_{k}}$ in non-decreasing during the whole transformation process, as 
   the ratio between the value after and before a step is

         \[\frac{\binom{\ell}{x_0, \dots, x_a-1,\dots,x_b+1,\dots,x_{k}} }{\binom{\ell}{x_0, \dots, x_a,\dots,x_b,\dots,x_{k}} }
        =\frac{x_a}{x_b+1}
       \ge \frac{z_a+1}{z_b}\ge1,\]
       where the last inequality follows from \eqref{eq:almost_decreasing}. 
   

       It is clear that the step cannot be applied infinitely, so eventually we arrive at $x$ for which the step is not defined, that is, $x=z$.
This finishes the proof of the proposition.
    \end{proof}

We now bound the multinomial coefficient $\binom{\ell}{z_0, \dots, z_{k}}$ from above by 
\[
  e^{O(\sqrt{k})}(\ell/\sqrt{k})^{\sqrt{k}}.
\]
Since we can assume $k$ is larger than some absolute constant, we can also assume that $z_0 \ge 1$.
By Stirling's formula, we have
$\ell!\le e\sqrt{\ell} (\ell/e)^\ell$ and $z_i!\ge \sqrt{2\pi z_i}(z_i/e)^{z_i}$. Using the latter estimate for all $i$ with $z_i \ge 1$ and the equality $z_i!=1$ otherwise, it follows that
\begin{equation*}
    \binom{\ell}{z_0, \dots, z_{k}}
    \le \frac{e\ell^\ell\sqrt{\ell}}{\displaystyle\prod_{i \ge 0 :z_i \ge 1}z_i^{z_i}\sqrt{2\pi z_i} }
    = O\left(\frac{\ell^\ell}{z_0^{z_0}}\cdot \frac{1}{\prod_{i\ge 1 : z_i \ge 1}z_i^{z_i}} \right),
\end{equation*}
where we used $\ell \le \prod_{i : z_i \ge 1}(2z_i)$, which can be shown by induction.
Using $1+x\le e^x$ for $x = \frac{\ell - z_0}{z_0} = \frac{\lfloor\sqrt{k}\rfloor}{z_0}$, we get
\begin{equation*}
    \frac{\ell^\ell}{z_0^{z_0}} 
    = \left(\frac{\ell}{z_0}\right)^{z_0}\ell ^{\left\lfloor\sqrt{k}\right\rfloor} \le e^{\left\lfloor\sqrt{k}\right\rfloor}\ell ^{\left\lfloor\sqrt{k}\right\rfloor}.
\end{equation*}
Let $i_k \left(\sim (k/4)^{1/3}\right)$ be the largest integer $i$ such that $\sqrt{k/i}-\sqrt{k/(i+1)}>1$.
Then using $x - 1 < \lfloor x \rfloor \le x$ it is routine to check that $z_i \ge 1$ for $i \le i_k$ and $z_i \le 1$ for $i > i_k$. 
Therefore 
\begin{align*}
    \frac1{\prod_{i \ge 1 : z_i \ge 1}z_i^{z_i}} 
    &= \frac1{   \prod_{i=1}^{i_k} z_i^{ z_i } }
     = \frac1{   \prod_{i=1}^{i_k} \left(\Omega\left( \frac{\sqrt{k}}{i^{3/2}} \right)\right)^{ z_i } } \\
    &= \left(\frac{1}{\sqrt{k}}\right)^{z_1+\cdots+z_{i_k}}
      \prod_{i=1}^{i_k} \left( O\left( i^{3/2} \right)\right)^{ O\left( \frac{\sqrt{k}}{i^{3/2}} \right)} \\
     &= \left(\frac{1}{\sqrt{k}}\right)^{\left\lfloor\sqrt{k}\right\rfloor-\left\lfloor\sqrt{k/(i_k+1)}\right\rfloor}
       e^{ O\left( \sum_{i=1}^{i_k}\frac{\sqrt{k}}{i^{3/2}} \ln i\right)} \\
      &= \left(\frac{1}{\sqrt{k}}\right)^{\left\lfloor\sqrt{k}\right\rfloor-\Theta(k^{1/3})}
      e^{ O\left(\sqrt{k} \right)}
      = \left(\frac{1}{\sqrt{k}}\right)^{\left\lfloor\sqrt{k}\right\rfloor}
      e^{ O\left(\sqrt{k} \right)}.
\end{align*}
      Putting the previous estimates together, we get
\begin{equation*}
   \binom{\ell}{z_0, \dots, z_{k}}
    \le  \left(\frac{\ell}{\sqrt{k}}\right)^{\sqrt{k}}
      e^{ O\left(\sqrt{k} \right)}.
\end{equation*}

   \subsection{Number of vectors $m$}
   Here we give an upper bound on the number of vectors $m = (m_0, \dots, m_k)$ with nonnegative integer coordinates satisfying~\eqref{eq:m_sum} and~\eqref{eq:m_tail}.
 We set $\beta:=\left\lceil k^{1/3} \right\rceil$. 
 By~\eqref{eq:m_tail}, $m_{\ge \beta} = m_{\beta} + \dots + m_{k} \le \sqrt{k/\beta} \le \beta$, so the number of ways to choose $m_\beta, \dots, m_k$ is at most the number of ways to put $m_{\ge \beta} \le \beta$ balls into $k - \beta + 1 \le k$ bins, which is at most $k^{\beta} = e^{O(k^{1/3} \ln k)}$.
 Since by \eqref{eq:m_sum} $m_i \in [0, \ell]$ for every $i$, the number of ways to choose $m_0,\dots,m_{\beta-1}$ is at most $(\ell + 1)^{\beta} = e^{O(k^{1/3}\ln \ell)}$. We conclude that the number of ways to choose $m$ is $e^{O(k^{1/3}(\ln \ell + \ln k))}$.
}
\section{Proof of Theorem~\ref{thm:n_two_selfint}}
\label{sec:proof_n_two_selfint}
\todo{SOLVED. Rev1 Proof of Theorem 1: it might be surprising that you can prove an upper bound on f(n,k) in such a way that you consider only self-intersections, and not pairwise intersections. Some explanation would be very useful. \emph{Matas: it is not completely true that we use only self-intersections to bound $f(n,k)$ (we do that to bound $g(n,k)$); when we apply Lemma~\ref{lem:snail_pandemonium} in the proof of Proposition~\ref{prop:ineqs}, we use that no pair of loops has $k$ intersections.} \resp{We added a sentence before Theorem~\ref{thm:n_two_selfint}} }
  Recall that without loss of generality assume $v = v_1$ so that the gaps adjacent to $v$ are $0$ and $1$. By Lemma~\ref{lem:reduced_words} we can assume that every $v$-loop in the collection induces an $\alpha\alpha$-free word so that the first and last letter is $2$.
  By Lemma~\ref{lem:nonhomotopicwords} \ref{en:nonhomotopicwords:noprefix}, the $v$-loops induce different words,  so it is enough to show that the number of words of the such form with self-intersection numbers less than $k$ is $e^{O(\sqrt{k})}$.

  Recall that the inner words use letters in $\{0, 1, 2\}$. Whenever we talk about two distinct letters $a,b$, let $c$ refer to the remaining third letter. 
  The maximal $ab$-words are disjoint and each of them is surrounded by $c$ or $v$. We can replace $w$ by an $\alpha\beta\alpha\beta$-free word $w'$  by repeatedly applying operation which replaces a subword of a form $abab$, for some distinct letters $a$ and $b$, by a subword $ab$. Note that this does not change the structure of maximal words and when the procedure terminates, every maximal word is a $\alpha\beta$-word or a $\alpha\beta\alpha$-word. Moreover, $w'$ remains $\alpha\alpha$-free and the first and last letters remain $2$ in each intermediate word.


  As discussed in subsection~\ref{ss:expansions}, word $w$ can be reconstructed from a $\alpha\beta\alpha\beta$-free word $w'$ by three consecutive $ab$-expansions, one for each $ab \in \left\{ 01, 02, 12 \right\}$. Note that we can also assume that, as in $w$, the first and the last letters in both $w'$ and the intermediate expansions are $2$. 

  In view of this claim, it is enough to count $\alpha\beta\alpha\beta$-free words with fewer than $k$ self-intersections, calculate the bound $\ell$ on the number of maximal $ab$-words in such a word, and apply Lemma~\ref{lem:expansions} three times.
  
  We claim that an $\alpha\alpha$-free word with fewer than $k$ self-intersections has at most $4\sqrt{k}$ maximal $ab$-words for each pair $ab$. Assume the contrary. Let $v$ be the obstacle incident to gaps $a$ and $b$. If at least $\sqrt{k}$ of the maximal $ab$-words are $s$-windings around $v$ with $s \ge 1$, then there are at least $k$ self-intersections by \eqref{eq:self_int_upper}. So further we assume there are at least $3\sqrt{k}$ maximal $ab$-words of the form $cabc$ or $cbac$. Note that each such word corresponds either to a $cab$-upsegment or a $cab$-downsegment. If among them there are at least $\sqrt{k}$ of each polarity, then we again have $k$ self-intersections between $cab$-upsegments and $cab$-downsegments by Lemma~\ref{lem:crossing_orient}. So further assume there are more than $2\sqrt{k}$ words $cabc$ of the same polarity (note that reversing $cabc$ does not change the polarity). In particular there is a pair of such maximal $ab$-words which has no other maximal $ab$ word between them. Since their polarity is the same, between them there is an even-length word which starts and ends in a letter other than $c$. We claim that such subword contains $ab$, giving a contradiction. Say the word has $2m$ letters and argue by induction: if $m = 1$, the first and last letter is not $c$, so the word is $ab$ or $ba$. Otherwise either the first two letters are $ab$ or $ba$ (in which case we're done) or one of $ac$ and $bc$, in which case removing them we are left with a shorter word of even number of letters starting and ending not in $c$, which by induction hypothesis contains $ab$.
 
  We have shown that every word $w$ with fewer than $k$ self-intersections is a `triple' extension of a $\alpha\beta\alpha\beta$-free word $w'$ with at most $\ell := 4 \sqrt{k}$ maximal $ab$-words for each $ab \in \{01,02,12\}$. Since the first maximal word in $w'$ has at most three letters, and each subsequent maximal word has at most two additional letters, $w'$ has at most $3 + 2(3\ell - 1) = 6\ell + 1$ letters.
  Taking into account that $w'$ is $\alpha\alpha$-free and necessarily starts with $2$, there are at most 
  \begin{equation}
    \label{eq:count_abab_free}
    \sum_{i = 2}^{6\ell + 1} 2^{i-1} \le 2^{6\ell} = e^{O(\sqrt{k})}
  \end{equation}
  choices of $w'$. 
  Now Lemma~\ref{lem:expansions} implies that the number of $ab$-expansions is 
 $e^{ O(\sqrt{k})},$
  which, together with \eqref{eq:count_abab_free} implies that the number of different words induced by the $v$-loops is $e^{ O(\sqrt{k})}$. By the remark at the beginning of the proof, this proves the theorem. \qed 

\section{Proof of Proposition~\ref{prop:ineqs}}
  \label{sec:prop:ineqs}
\ifarxivversion
  For the full proof, see Appendix~\ref{apx:prop_ineqs}.
\fi

  \paragraph{Proof sketch.} 
  
  Let $f = f(n,k)$ and choose a family of non-homotopic $x$-loops $\ell_1, \dots, \ell_{f}$ in $S$ of the maximal size.
  We choose 
  an obstacle $v$, 
  a point $x'$ on some loop, 
  and a path $P$ from $v$ to $x'$ that does not intersect any loop. We assume without loss of generality that $x'$ lies on the loop $\ell_f$.

  We first turn the $x$-loops into non-homotopic $x'$-loops while keeping the number of pairwise intersections and self-intersections bounded.
  We choose a path $R$ with no self-intersections that connects $x$ to $x'$ and is contained in the graph of $\ell_f$.
  The $x$-loop $\ell_f$ is already an  $x'$-loop and we turn every other $x$-loop $\ell_i$ into an $x'$-loop $\ell'_i$ by (i) following $R$ from $x'$ to $x$, (ii) going along $\ell_i$, and (iii) returning back to $x'$ via $R$. Note that to avoid an infinite number of (self-)intersections, the loops cannot follow $R$ precisely. Rather, they use pairwise disjoint paths that run along $R$ in sufficiently small distance so that if any loop intersects them it must also intersect $R$ (and thus $\ell_f$). See the left and middle parts of Figure~\ref{fig:proposition_ineqs}.
  
  \begin{figure}[h]
    \centering
    \includegraphics[width=1.0\textwidth]{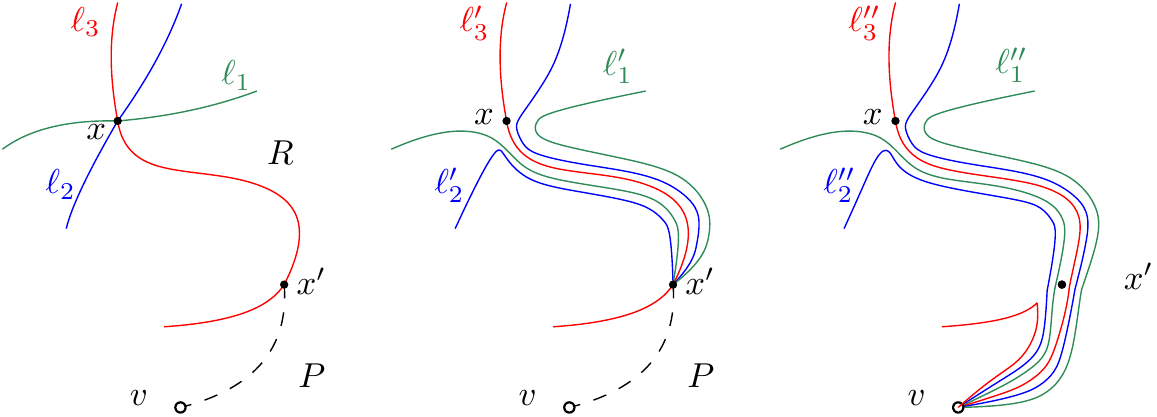}
    \caption{Transforming $x$-loops into $x'$-loops and then to $v$-loops.}%
    \label{fig:proposition_ineqs}
  \end{figure}
   
  It is easy to see that the obtained $x'$-loops are pairwise non-homotopic. And since $R$ is a subset of the loop $\ell_f$, every newly created crossing between a pair of loops $\ell'_i$ and $\ell'_j$ corresponds to some crossing between $\ell_i$ and $\ell_f$, or $\ell_j$ and $\ell_f$. This fact allows us to bound the total number of additional (self-)intersections by $4k$.
  
  \begin{figure}[h]
    \centering
    \includegraphics[width=.30\textwidth]{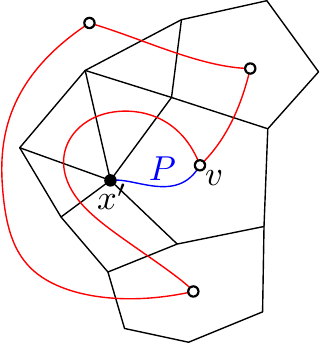}
    \caption{Drawing equator so that it does not separate $x'$ from the obstacle $v$.}%
    \label{fig:drawing_equator}
  \end{figure}

  Now we choose the equator so that it does not cross the path $P$ (see Figure~\ref{fig:drawing_equator}). 
  Applying Lemma~\ref{lem:reduced_words}.\ref{enum:x_not_obst} we modify the $x'$-loops without increasing the numbers of intersections so that they induce $\alpha\alpha$-free words.
  
  Finally we turn each $x'$-loop $\ell'_i$ into a $v$-loop $\ell''_i$ so that no additional intersections are created, and the inner word that $\ell''_i$ induces is the same as the word induced by $\ell'_i$. This is done similarly as before -- the loop $\ell''_i$ is obtained by concatenating (i) a path closely following $P$ from $v$ to $x'$, (ii) the $x'$-loop $\ell'_i$, and (iii) a path closely following $P$ back to $v$. See the right part of Figure~\ref{fig:proposition_ineqs}.

  Some resulting $v$-loops may be homotopic. Partition the $v$-loops into maximal sets of homotopic $v$-loops $H_1, \dots, H_m$ and note that $m \le g(n, 5k)$. Since the first and the last arc of every $v$-loop lies in the same hemisphere, by Lemma~\ref{lem:nonhomotopicwords} \ref{en:nonhomotopicwords:samehemi} for each set $H_j$ there is a word $u_j$ starting in letters other than $0$ or $1$ so that each $\ell \in H_j$ induces an even-length word of the form 
  \[
  v w'_\ell u_j w''_\ell  v, 
  \]
  where words $w'_\ell$ and $w''_\ell$ use letters $0$ and $1$.
  Applying Lemma~\ref{lem:snail_pandemonium}, we see that $|H_j| \le 4(2 \cdot 5k + 1)^2 \le 4( 11k)^2 = 484 k^2$ for every $j$. And since $m \le g(n,5k)$, this implies that $f(n,k) \le 484k^2 g(n,5k)$.
 
 \toappendix{
 \section{Full proof of Proposition~\ref{prop:ineqs}}
 \label{apx:prop_ineqs}
  \begin{proof}[Proposition~\ref{prop:ineqs}]
   Let $f = f(n,k)$ and choose a family $L$ of non-homotopic $x$-loops $\ell_1, \dots, \ell_{f}$ in $S$, attaining the maximum in the definition of $f(n,k)$. Removing the image of all loops partitions $\R^2$ into connected open sets, among which we pick $F$ that contains some obstacle $v$ on its boundary. Moreover, we choose a point $x'$ in the boundary $\partial F \setminus V_n$, and a path $P$ connecting $v$ to $x'$ inside $F$. We assume without loss of generality that $x'$ lies on the loop $\ell_f$.
    
   Let $R$ be a path without self-intersections that is contained in the graph of the the loop $\ell_f$ and connects $x$ to $x'$. (We obtain $R$ simply by taking one of the two segments of $\ell_f$ between $x$ and $x'$ and deleting every loop.) Moreover, let us fix a circle $C$ centered at $x$ small enough such that every loop $\ell_i$ intersects $C$ exactly twice (at the very beginning and at the very end of the loop). Denote these points as $a_{2i - 1}$ and $a_{2i}$. Since the path $R$ intersects $C$ at some point $y$ that is different from all $a_i$s, we can choose points $y_1, \dots, y_{2f}$ on $C$ close to $b$ and connect them to $x'$ using disjoint paths $R_1, \dots, R_{2f}$ that closely follow the original path $R$. Moreover, we can choose the order of the points $y_i$ so that that the straight segments $a_iy_i$ are all disjoint.
  
  \begin{figure}[h]
    \centering
    \includegraphics[width=1.0\textwidth]{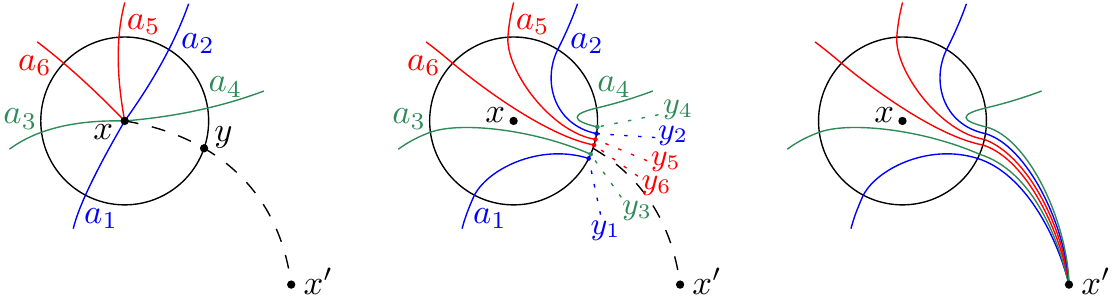}
    \caption{Transforming $x$-loops into $x'$-loops.}%
    \label{fig:x_to_x_prime}
  \end{figure}
  
  We define a family $L'$ of $x'$-loops $\ell'_1, \dots, \ell'_{f}$ where $\ell'_i$ for $i \in [f-1]$ is obtained by concatenating (i) the path $R_{2i-1}$ from $x'$ to $y_{2i-1}$, (ii) the segment $y_{2i-1}a_{2i-1}$, (iii) the part of $\ell_i$ between $a_{2i-1}$ and $a_{2i}$, (iv) the segment $a_{2i}y_{2i}$, and (v) the path $R_{2i}$ back to $x'$. The part from $x'$ to $a_{2i-1}$ is called the \emph{head} of $\ell'_i$, the part from $a_{2i}$ back to $x'$ is called the \emph{tail} of $\ell'_i$ and the remaining middle part is called the \emph{body} of $\ell'_i$. Finally, we set $\ell'_f = \ell_f$ as $\ell_f$ already passes through $x'$.  See Figure~\ref{fig:x_to_x_prime}.
  
  First, we show that the $x'$-loops in $L'$ are pairwise non-homotopic. Suppose for a contradiction that $x'$-loops $\ell'_i$ and $\ell'_j$ are homotopic. For any $h \in [f]$, let $\ell^R_h$ be the $x$-loop obtained by following $R$ from $x$ to $x'$ then going along the $x'$-loop $\ell'_h$ and finally returning to $x$ via $R$. Clearly, we have $\ell^R_h \sim \ell_h$ for every $h$. Furthermore, the loops $\ell^R_i$ and $\ell^R_j$ can be shown to be homotopic by applying the homotopy between $\ell'_i$ and $\ell'_j$ to the middle parts of $\ell^R_i$ and $\ell^R_j$. Therefore,
  $\ell_i \sim \ell^R_i \sim \ell^R_j \sim \ell_j$ which contradicts that $L$ is a family of non-homotopic loops.
  
  Now we show that every loop in $L'$ has fewer than $5k$ self-intersections and any two loops have fewer than $5k$ intersections. To that end, let $\ell'_i$ and $\ell'_j$ be two loops from $L'$, not necessarily different. Firstly, we have at most $k$ intersections between $\ell_i$ and $\ell_j$ (and thus between the bodies of $\ell'_i$ and $\ell'_j$) so it remains to bound only the intersections incident with heads and tails of $\ell'_i$ and $\ell'_j$. The paths $R_{2i-1}$, $R_{2i}$, $R_{2j-1}$ and $R_{2j}$ are pairwise disjoint and without self-intersections. Therefore, the only newly added intersections are between the head or the tail of $\ell'_i$ and the body of $\ell'_j$, and vice versa. Since $R$ is a subset of the loop $\ell_f$, there are at most $k$ intersections between $R$ and any $x$-loop $\ell_h \in L$ and thus at most $k$ newly added intersections between the head of $\ell'_i$ and the body of $\ell'_j$. Counting the symmetric situations, that gives at most $4k$ additional intersections. However, we need to be more careful when $\ell'_j$ is taken to be $\ell'_f$. In such case, any self-intersection of $\ell_f$ where we shortened a loop during the construction of $R$ can create two new intersections with the head or tail of $\ell'_i$. But since $\ell'_f$ itself has no head or tail, there are only at most $2k$ new intersections incident with the head of $\ell'_i$ and $2k$ new intersections incident with the tail of $\ell'_i$.
  

   Now we choose the equator so that it does not cross the path $P$
   (see Figure~\ref{fig:drawing_equator}). 
   Applying Lemma~\ref{lem:reduced_words}.\ref{enum:x_not_obst} we modify the $x'$-loops without increasing the numbers of intersections so that they induce $\alpha\alpha$-free words. 

   \begin{figure}[h]
     \centering
     \includegraphics[width=1.0\textwidth]{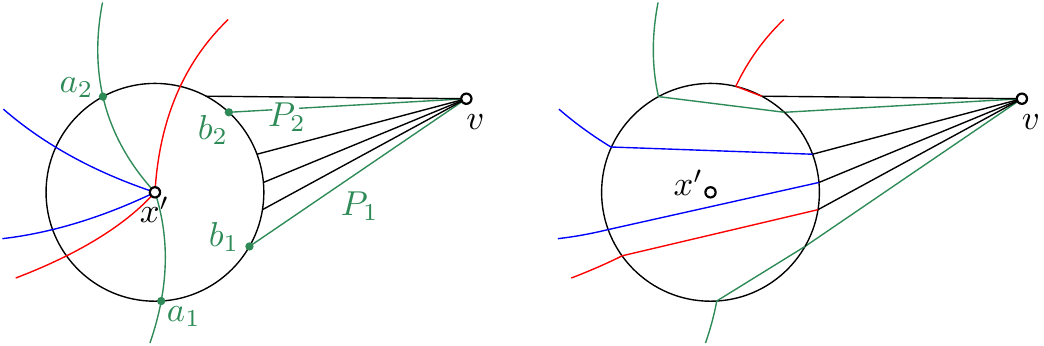}
     \caption{Transforming an $x'$-loop into a $v$-loop.}%
     \label{fig:move_x_to_v}
   \end{figure}
   Finally we turn each $x'$-loop $\ell'_i$ into a $v$-loop $\ell''_i$ so that no additional intersections are created, and the inner word that $\ell''_i$ induces is the same as the word induced by $\ell'_i$.
   If we fix a circle $C'$ centered at $x'$ that is small enough, then every $x'$-loop $\ell'_i$ intersects it exactly twice (here we use the assumption that $x'$-loops do not pass through $x'$). Denote these points as $a_{2i - 1}$ and $a_{2i}$.
   The equator partitions the face $F$ into several connected components. Let $F'$ be the one which contains the path $P$.
   Since the path $P$ intersects $C'$ at some point $b$ that is different from all $a_i$s, we can choose points $b_1, b_{2f}$ on $C'$ close to $b$ and connect them to $v$ using disjoint paths $P_1, \dots, P_{2f}$ inside $F'$ (and thus not intersecting the equator or any loop).
   Moreover, we can choose the order of the points $b_i$ so that that the straight segments $a_ib_i$ are all disjoint.
   Now for each $\ell'_i$ construct $\ell''_i$ by concatenating (i) the path $P_{2i - 1}$, (ii) the straight segment $b_{2i - 1}a_{2i-1}$, (iii) the part of $x'$-loop $\ell'_i$ outside of the circle, (iv) straight segment $a_{2i}b_{2i}$, and (v) the path $P_{2i}$, see Figure~\ref{fig:move_x_to_v}.

  Some resulting $v$-loops may be homotopic.
  Partition the $v$-loops $L'' = H_1 \cup \dots \cup H_m$ into nonempty sets accoding to the homotopy class and note that $m \le g(n, 5k)$. Since the first and the last arc of every $v$-loop lies in the same hemisphere, by Lemma~\ref{lem:nonhomotopicwords} \ref{en:nonhomotopicwords:samehemi} for each set $H_j$ there is a word $u_j$ starting in letters other than $0$ or $1$ so that each $\ell \in H_j$ induces an even-length word of the form 
\[
  v w'_\ell u_j w''_\ell  v, 
\]
where words $w'_\ell$ and $w''_\ell$ use letters $0$ and $1$.
By Lemma~\ref{lem:snail_pandemonium}, $|H_j| \le 4(2 \cdot 5k + 1)^2 \le 4( 11k)^2 < 484 k^2$ for every~$j$.

   
Since $m \le g(n,5k)$, this implies that $f(n,k) \le 484k^2 g(n,5k)$.
  This completes the proof of Proposition~\ref{prop:ineqs}.
  \end{proof}
}




\bibliographystyle{splncs04}
\bibliography{bibliography}

\begin{thebibliography}{1}
\providecommand{\url}[1]{\texttt{#1}}
\providecommand{\urlprefix}{URL }
\providecommand{\doi}[1]{https://doi.org/#1}

\bibitem{ACNS}
Ajtai, M., Chv\'{a}tal, V., Newborn, M.M., Szemer\'{e}di, E.: Crossing-free
  subgraphs. In: Theory and practice of combinatorics, North-Holland
  Mathematics Studies, vol.~60, pp. 9--12. North-Holland, Amsterdam (1982)

\bibitem{BOSV_CALDAM}
Bla{\v{z}}ej, V., Opler, M., {\v{S}}ileikis, M., Valtr, P.: On the
  intersections of non-homotopic loops. In: Conference on Algorithms and
  Discrete Applied Mathematics. pp. 196--205. Springer (2021)

\bibitem{JMM}
Juvan, M., Malni\v{c}, A., Mohar, B.: Systems of curves on surfaces. J. Combin.
  Theory Ser. B  \textbf{68}(1),  7--22 (1996). \doi{10.1006/jctb.1996.0053}

\bibitem{Leighton}
Leighton, F.T.: Complexity Issues in {VLSI}. Foundations of Computing, MIT
  Press, Cambridge, MA, USA (1983)

\bibitem{PTT}
Pach, J., Tardos, G., T{\'{o}}th, G.: Crossings between non-homotopic edges.
  In: Auber, D., Valtr, P. (eds.) Graph Drawing and Network Visualization -
  28th International Symposium, {GD} 2020, Vancouver, BC, Canada, September
  16-18, 2020, Revised Selected Papers. Lecture Notes in Computer Science, vol.
  12590, pp. 359--371. Springer (2020)

\end{thebibliography}

\ifarxivversion 
  \newpage 
  \appendix

  \section{Appendix} 
  \appendixText 
\fi

\end{document}